\documentclass[a4paper,USenglish,cleveref, autoref, thm-restate]{lipics-v2021}

\usepackage{amssymb}
\usepackage{amsmath}
\usepackage{comment}
\usepackage[disable]{todonotes}
\usepackage[utf8]{inputenc}
\usepackage{shuffle}
\usepackage{multirow}
\usepackage{mathabx}
\usepackage{mathrsfs}
\usepackage{wasysym}

\usepackage{apxproof}
\usepackage{xspace}

\newtheoremrep{theorem}{Theorem}[section]
\newtheoremrep{proposition}[theorem]{Proposition}
\newtheoremrep{lemma}[theorem]{Lemma}
\newtheoremrep{claim}[theorem]{Claim}
\newtheoremrep{conjecture}[theorem]{Conjecture}
\newtheoremrep{corollary}[theorem]{Corollary}
\theoremstyle{definition}
\newtheoremrep{definition}[theorem]{Definition}
\newtheoremrep{example}[theorem]{Example}

\usepackage{hyperref}

\usepackage{tikz}
\usetikzlibrary{arrows,shapes,automata,positioning}
\usetikzlibrary{arrows.meta,bending,chains,matrix,positioning,scopes,calc,automata, shapes, patterns}
\tikzstyle{transition} = [font=\small]
\tikzstyle{line} = [draw,thick,-latex]
\usetikzlibrary{fadings,patterns}
\usetikzlibrary{intersections}
\usetikzlibrary{decorations.pathreplacing,calligraphy}
\pgfdeclarelayer{background}
\pgfsetlayers{background,main}

\newcommand{\NP}{\textsf{NP}}
\newcommand{\PSPACE}{\textsf{PSPACE}}

\newcommand{\PTIME}{\textsf{P}}
\newcommand{\XP}{\textsf{XP}}

\usepackage{tcolorbox}

\usepackage{tabularx,lipsum,environ,amsmath,amssymb}

\newcounter{problemcounter}
\makeatletter
\newcommand{\problemtitle}[1]{\gdef\@problemtitle{#1}}
\newcommand{\probleminput}[1]{\gdef\@probleminput{#1}}
\newcommand{\problemquestion}[1]{\gdef\@problemquestion{#1}}
\NewEnviron{decproblem}{
  \refstepcounter{problemcounter}
  \problemtitle{}\probleminput{}\problemquestion{}
  \BODY
  \par\addvspace{.5\baselineskip}
  \noindent
  \begin{tabularx}{\textwidth}{@{\hspace{\parindent}} l X c}
    \multicolumn{2}{@{\hspace{\parindent}}l}{\textbf{Decision Problem \theproblemcounter:} \@problemtitle} \\
    \textbf{Input:} & \@probleminput \\
    \textbf{Question:} & \@problemquestion
  \end{tabularx}
  \par\addvspace{.5\baselineskip}
}
\makeatother

\nolinenumbers

\title{Constrained Synchronization for Commutative Automata and Automata with Simple Idempotents}

\titlerunning{Synchronization for Commutative Aut. and Aut. with Simple Idempotents}

\author{Stefan Hoffmann}
{Informatikwissenschaften, FB IV, Universit\"at Trier, Germany}
{hoffmanns@informatik.uni-trier.de}
{https://orcid.org/0000-0002-7866-075X}
{}

\authorrunning{S. Hoffmann} 

\Copyright{Stefan Hoffmann} 

\ccsdesc[500]{Theory of computation~Regular languages}
\ccsdesc[500]{Theory of computation~Problems, reductions and completeness}

\keywords{Constrained Synchronization,Commutative Automata,Automata with Simple Idempotents}

\EventEditors{John Q. Open and Joan R. Access}
\EventNoEds{2}
\EventLongTitle{42nd Conference on Very Important Topics (CVIT 2016)}
\EventShortTitle{CVIT 2016}
\EventAcronym{CVIT}
\EventYear{2016}
\EventDate{December 24--27, 2016}
\EventLocation{Little Whinging, United Kingdom}
\EventLogo{}
\SeriesVolume{42}
\ArticleNo{23}

\begin{document}
\maketitle
\begin{abstract}
 For general input automata, there exist regular constraint languages
 such that asking if a given input automaton admits
 a synchronizing word in the constraint language
 is \PSPACE-complete or \NP-complete.
 Here, we investigate this problem for commutative automata over an arbitrary alphabet
 and automata with simple idempotents over a binary alphabet
 as input automata. The latter class contains, for example, the \v{C}ern\'y family of automata.
 We find that for commutative input automata, the problem
 is always solvable in polynomial time, for every constraint language.
 For input automata with simple idempotents over a binary alphabet and with a
 constraint language given by a partial automaton with up to three states,
 the constrained synchronization problem is also solvable in polynomial time.
\end{abstract}

\todo[inline]{Noch Subset synchronizatino probleme angucken, monotone automaten

meine FCT, COCOON, ISCTC paper zitieren?

noch solvable ergebnisse aufnehmen?

annahme eingabe über gleichem sigma, sonst z.b. circular hart wenn man zyklus-buchstaben ignoriert

nummerierung der theorem/defs etc

csp als abkürzung rausnehmen

|G| teilt n!, wenn kardinalität nicht geänadert und idempotent 2 angewant, dann 2x die gleiche menge.
|P||G|

alle simple idempotent -> O(1) [kleinste menge beschränkt, beispiel i.A. nicht bei zwei rank n - 1 die immer hin- und hershiften], aber für idempotents n - (|Sigma|+1) mengen erreichbar
https://arxiv.org/pdf/2011.14404.pdf (S.13)

genau eine permutation
genau eine idempotent

nur bounded7sparse constraints und smple idempotent automata, NP möglich?

 bilder einfügen simple idempotent cases. absract neu schreiben. das zeug aus conclusion in haupttext einstreuen.
 
 bild cerny automat? oder bei beschriftung vno obigen bild hinweisen, dass sich dieser so ergibt.

 ich zitieren [30] meine arxiv version, kann doch lata zitieren..., generell viele referenzen doppelt (hopcroft), referenzen reduzieren?

 abstract umformulieren

}

\section{Introduction}

A deterministic semi-automaton (which is an automaton without a distinguished start state and without a set of final states) is \emph{synchronizing} if it admits a reset word, i.e., a word which leads to a definite
state, regardless of the starting state. This notion has a wide range of applications, from software testing, circuit synthesis, communication engineering and the like, see~\cite{DBLP:journals/et/ChoJSP93,San2005,Vol2008}.

The famous \v{C}ern\'y conjecture \cite{Cerny64}
states that a minimal synchronizing word, for an $n$ state automaton, has length
at most $(n-1)^2$. 
The best general upper bound known so far is cubic~\cite{shitov2019}. Historically, the following
bounds have been published:

\begin{center}
\begin{tabular}{ l@{\hskip 0.2in} l@{\hskip 0.2in} l }
    $2^n - n - 1$ & & (1964, \v{C}ern\'y~\cite{Cerny64}) \\ [2pt]
    $\frac{1}{2}n^3 - \frac{3}{2} n^2 + n + 1$ & & (1966, Starke~\cite{Starke66}) \\   [2pt]
    $\frac{1}{2} n^3 - n^2 + \frac{n}{2}$ & & (1970, Kohavi~\cite{Kohavi70}) \\  [2pt]
    $\frac{1}{3} n^3 - n^2 - \frac{1}{3} n + 6$ & & (1970, Kfourny~\cite{Kfoury70}) \\ [2pt]
    $\frac{1}{3} n^3 - \frac{3}{2} n^2 + \frac{25}{6} n - 4$ & & (1971, \v{C}ern\'y et al.~\cite{CernyPirickaRosenauerova71}) \\ [2pt]
    $\frac{7}{27} n^3 - \frac{17}{18} n^2 + \frac{17}{6} n - 3$ & $n \equiv 0 \pmod{3}$ & (1977, Pin~\cite{Pin77a}) \\ [2pt]
    $\left(\frac{1}{2} + \frac{\pi}{36}\right) n^3 + o(n^3)$ & & (1981, Pin~\cite{Pin81a}) \\ [2pt]
    $\frac{1}{6}n^3 - \frac{1}{6}n - 1$ & & (1983, Pin/Frankl~\cite{Frankl82,Pin83a}) \\ [2pt]
 $\alpha n^3 + o(n^3)$ & $\alpha \approx 0.1664$  & (2018, Szyku\l a~\cite{szykula2018}) \\ [2pt]
 $\alpha n^3 + o(n^3)$ & $\alpha \le 0.1654$ & (2019, Shitov~\cite{shitov2019})    
\end{tabular}
\end{center}

The \v{C}ern\'y conjecture \cite{Cerny64} has been confirmed for a variety
of classes of automata, for example: circular automata~\cite{Dubuc96,Dubuc98,Pin78a}, oriented or (generalized) monotonic automata~\cite{AnanichevVolkov03,AnanichevV05,Eppstein90} (even the better bound $n - 1$),
automata with a sink state~\cite{Rys96} (even the better bound $\frac{n(n-1)}{2}$), solvable and commutative automata~\cite{FernauHoffmann19,Rys96,Rystsov97} (even the better bound $n - 1$), weakly acyclic automata~\cite{DBLP:journals/tcs/Ryzhikov19a},
Eulerian automata~\cite{JKari01b}, automata preserving a chain of partial orders~\cite{Volkov09},
automata whose transition monoid contains a $\operatorname{\mathbb{Q} I}$-group~\cite{ArCamSt2017,ArnoldS06},
certain one-cluster automata~\cite{Steinberg11a},
automata that cannot recognize $\{a,b\}^*ab\{a,b\}^*$~\cite{AlmeidaS09}, aperiodic automata~\cite{Trahtman07} (even the better bound $\frac{n(n-1)}{2}$), certain aperiodically $1$-contracting automata~\cite{Don16}
and automata having letters of a certain rank~\cite{BerlinkovS16}.

Additionally, the bound $2(n-1)^2$ has been obtained for the following two classes: automata\todo{hier stand autoamata, also auch in DCFS und journal LATA falsch! punkte am ende fehlte auch.}
with simple idempotents~\cite{Rystsov2000}
and regular automata~\cite{Rystsov1995b,Rystsov1995a}.

For further information, we refer to the mentioned survey articles for details~\cite{San2005,Vol2008}.

Due to its importance, the notion of synchronization has undergone a range of generalizations and variations
for other automata models.
The paper~\cite{FernauGHHVW19} introduced the constrained synchronization problem. 
In this problem, we search for a synchronizing word coming from a specific subset of allowed
input sequences. 
To sketch a few applications:

%


\begin{description}

\item[Reset State.] 

In~\cite{FernauGHHVW19} one motivating example was the demand that a system, or automaton thereof, to synchronize has to first enter a ``directing'' mode, perform a sequence of
operations, and then has to leave this operating mode and enter the ``normal
operating mode'' again. In the most simple case, this constraint
can be modelled by $ab^*a$, which, as it turns out~\cite{FernauGHHVW19},
yields an \NP-complete constrained synchronization problem.
Even more generally, it might be possible that a system -- a remotely controlled
rover on a distant planet, a satellite in orbit, or a lost autonomous vehicle
-- is not allowed to execute all commands in every possible
order, but  certain commands are only allowed in certain order or after
other commands have been executed. All of this imposes constraints
on the possible reset sequences.

\item[Part Orienters.] Suppose parts arrive at a manufacturing site and they
need to be sorted and oriented before assembly. Practical considerations
favor methods which require little or no sensing, employ simple devices,
and are as robust as possible. This can be achieved as follows.
We put parts to be oriented on a conveyor belt which takes them to the assembly
point and let the stream of the parts encounter a series of passive obstacles placed along
the belt. Much research on synchronizing automata was motivated
by this application~\cite{DBLP:journals/algorithmica/ChenI95,Eppstein90,DBLP:journals/trob/ErdmannM88,DBLP:journals/algorithmica/Goldberg93,DBLP:conf/focs/Natarajan86,DBLP:journals/ijrr/Natarajan89,Vol2008}
and I refer to~\cite{Vol2008} for
an illustrative example. 
Now, furthermore, assume the passive components can not be placed at random along
the belt, but have to obey some restrictions, or restrictions
in what order they are allowed to happen. These can be due
to the availability of components, requirements how to lay things out
or physical restrictions. 

\item[Supervisory Control.] The constrained synchronization problem
 can also be viewed of as 
 supervisory control
 of a discrete event system (DES) that is given by an automaton
 and whose event sequence is modelled by a formal language~\cite{DBLP:books/daglib/0034521,RamadgeWonham87,wonham2019}.
 In this framework, a DES has a set of controllable
 and uncontrollable events.
 Dependent
 on the event sequence that occurred so far,
 the supervisor is able to restrict the set of events
 that are possible in the next step, where, however,
 he can only limit the use of controllable events.
 So, if we want to (globally) reset a finite state DES~\cite{Alves2020} under supervisory control, 
 this is equivalent to constrained synchronization problem.

\end{description}

In~\cite{FernauGHHVW19} it was shown that we can realize PSPACE-complete, NP-complete
or polynomial time solvable constrained problems by appropriately choosing a
constraint language. Investigating the reductions from~\cite{FernauGHHVW19}, 
we see that most
reductions yield automata with a sink state, which then must be the unique
synchronizing state. Hence, we can conclude that we can realize these complexities with this type of input automaton.
Contrary, for example, unary automata are synchronizing only if they admit
no non-trivial cycle, i.e., only a single self-loop. In this case, we can easily decide
synchronizability for any constraint language in polynomial time. Hence, for
these simple types of automata, the complexity drops considerably. So, a natural
question is, if we restrict the class of input automata, what complexities are
realizable? Or more precisely:

\begin{quote} 
 What features in the input automata do we need to realize certain complexities?
\end{quote}

In~\cite{DBLP:conf/dlt/Hoffmann21} this question was investigated for weakly acyclic, or partially ordered,
input automata. These are automata where all cycles are trivial, i.e., the only loops
are self-loops. It was shown that in this case, the constrained synchronization problem
is always in \NP\  and, for suitable constraint languages, \NP-complete problems
are realizable.


\subparagraph{Overview and Contribution}
We investigate the constrained synchronization when the input
is restricted to the class of commutative input automata
and to the class of input automata with simple idempotents. Both classes
were investigated previously with respect to the \v{C}ern\'y conjecture~\cite{FernauHoffmann19,Rys96,Rystsov2000} (see the list of the automata classes above for which this conjecture has been confirmed) and with respect
to computational problems of computing a shortest synchronizing word~\cite{Mar2009b}.
We show in Section~\ref{sec:comm_input} that for commutative input automata over an arbitrary alphabet and arbitrary constraint automata,
and in Section~\ref{sec:simpl_idemp} that for input automata with simple idempotents over a binary alphabet
and small constraint automata, the constrained
synchronization problem is always solvable in polynomial time.

Section~\ref{sec:comm_input} splits
into three subsections.
The first, Subsection~\ref{subsec:structure_sync_words}, is concerned with the set of synchronizing words for commutative semi-automata.
In Subsection~\ref{subsec:recognizing_intersection},
we show an auxiliary result that is also of independent
interest, namely that, given $m$ weakly acyclic
and commutative automata, the set of words accepted
by them all is recognizable by a weakly acyclic automaton
computable in polynomial time for a fixed alphabet.
Then, in Subsection~\ref{subsec:poly_alg},
we combine all these results
to show Theorem~\ref{thm:main_theorem_comm_input},
the statement that for commutative input automata
the constrained synchronization problem is tractable.

Section~\ref{sec:simpl_idemp} uses, stated in Proposition~\ref{prop:structure_simp_idem_binary},
that synchronizing automata with simple idempontent over a binary alphabet must have a very specific form.
Then, Theorem~\ref{thm:simple_idempotents_binary_in_P},
uses this result to give a polynomial time algorithm
for every possible constraint automaton
over a binary alphabet with at most three states.

\section{Preliminaries and Some Known Results}
\label{sec:preliminaries}

We assume the reader to have some basic knowledge in computational complexity theory and formal language theory, as contained, e.g., in~\cite{HopUll79}. For instance, we make use of  regular expressions to describe languages.
By $\Sigma$ we denote the \emph{alphabet}, a finite set.
For a word $w \in \Sigma^*$ we denote by $|w|$ its \emph{length},
and, for a symbol $x \in \Sigma$, we write $|w|_x$ to denote the \emph{number of occurrences of $x$}
in the word. We denote the empty word, i.e., the word of length zero, by $\varepsilon$.
We call $u \in \Sigma^*$ a \emph{prefix} of a word $v \in \Sigma^*$
if there exists $w \in \Sigma^*$ such that $v = uw$.
For $U, V \subseteq \Sigma^*$, we set $U\cdot V = UV = \{ uv \mid u \in U, v \in V \}$
and 
$U^0 = \{ \varepsilon \}$, $U^{i+1} = U^i U$, 
and $U^* = \bigcup_{i \ge 0} U^i$ and $U^+ = \bigcup_{i > 0} U^i$.
We also make use of complexity classes like $\PTIME$, $\NP$, or $\PSPACE$.

A \emph{partial deterministic finite automaton (PDFA)} is a tuple $\mathcal A = (\Sigma, Q, \delta, q_0, F)$,
where $\Sigma$ is a finite set of \emph{input symbols},~$Q$ is the finite \emph{state set}, $q_0 \in Q$ the \emph{start state}, $F \subseteq Q$ the \emph{final state set} and $\delta \colon Q\times \Sigma \rightharpoonup Q$ the \emph{partial transition function}.
The \emph{partial transition function} $\delta \colon Q\times \Sigma \rightharpoonup Q$ extends to words from $\Sigma^*$ in the usual way. 
Furthermore, for $S \subseteq Q$ and $w \in \Sigma^*$, we set $\delta(S, w) = \{\,\delta(q, w) \mid \mbox{$\delta(q,w)$ is defined and } q \in S\,\}$.
We call $\mathcal A$ a \emph{complete (deterministic finite) automaton} if~$\delta$ is defined for every $(q,a)\in Q \times \Sigma$.
If $|\Sigma| = 1$, we call $\mathcal A$ a \emph{unary automaton} and
 $L \subseteq \Sigma^*$ is also called a \emph{unary language}.
The set $L(\mathcal A) = \{\, w \in \Sigma^* \mid \delta(q_0, w) \in F\,\}$ denotes the language
\emph{recognized} 
by~$\mathcal A$.

A \emph{deterministic and complete semi-automaton (DCSA)} $\mathcal A = (\Sigma, Q, \delta)$
is a deterministic and complete finite automaton without a specified start state
and with no specified set of final states.
When the context is clear, we call both deterministic finite automata and semi-automata simply \emph{automata}.
Here, when talking about semi-automata, we always mean complete and deterministic
semi-automata, as we do not consider other models of semi-automata.
\todo{das vielleicht rausnehmen?}
Concepts and notions that only rely on the transition structure
carry over from complete automata to semi-automata and vice versa
and we assume, for example, that every notions defined
for semi-automata has also the same meaning for complete automata
with a start state and a set of final states.

Let $\mathcal A = (\Sigma, Q, \delta)$ be a semi-automaton. A maximal subset $S \subseteq Q$
with the property that for every $s, t \in S$
there exists $u \in \Sigma^*$ such that $\delta(s, u) = t$
is called a \emph{strongly connected component} of $\mathcal A$.
We also say that a state $s \in Q$ is \emph{connected} to a state $t \in Q$ (or $t$ is \emph{reachable} from $s$) if there exists $u \in \Sigma^*$ such that $\delta(s, u) = t$.
Viewing the strongly connected components as vertices of a directed graph
with edges being induced by the transitions that connect different
components, we get an acyclic
directed graph.

Let $\mathcal A = (\Sigma, Q, \delta)$ be a semi-automaton.
Then, $\mathcal A$ is called an \emph{automaton with simple idempotents},
if every $a \in \Sigma$ either permutes the states or maps precisely two states to a single state
and every other state to itself. More formally, every $a \in \Sigma$
either (1) permutes the states, i.e., $\delta(Q, a) = Q$,
or (2) it is a \emph{simple idempotent}, i.e., we have
$|\delta(Q, a)| = |Q| - 1$ and $\delta(q, aa) = \delta(q, a)$
for every $q \in Q$.
Letters fulfilling condition (1) are also called \emph{permutational letters}, and letters fulfilling (2)
are called \emph{simple idempotent letters}.

The semi-automaton $\mathcal A$ is called \emph{commutative}, if for all $a,b \in \Sigma$
and $q \in Q$ we have $\delta(q, ab) = \delta(q, ba)$.
The semi-automaton $\mathcal A$ is called \emph{weakly acyclic}, if there exists an ordering
$q_1, q_2, \ldots, q_n$ of its states such that if $\delta(q_i, a) = q_j$
for some letter $a \in \Sigma$, then $i \le j$ (such an ordering is called
a \emph{topological sorting}).
An automaton is weakly acyclic, if $\delta(q, uxv) = q$
for $u, v \in \Sigma^*$ and $x \in \Sigma$
implies $\delta(q, x) = q$, i.e., the only loops in the automaton
graph are self-loops.
This is also equivalent
to the fact that the reachability relation between the
states is a partial order.


A complete automaton $\mathcal A$ is called \emph{synchronizing} if there exists a word $w \in \Sigma^*$ with $|\delta(Q, w)| = 1$. In this case, we call $w$ a \emph{synchronizing word} for $\mathcal A$.
We call a state $q\in Q$ with $\delta(Q, w)=\{q\}$ for some $w\in \Sigma^*$ a \emph{synchronizing state}.
For a semi-automaton (or PDFA) with state set $Q$ and transition function $\delta : Q \times \Sigma \rightharpoonup Q$,
a state $q$ is called a \emph{sink state}, if for all $x \in \Sigma$ we have $\delta(q,x) = q$.
Note that, if a synchronizing automaton has a sink state, then the
synchronizing state is unique and must equal the sink state.

In~\cite{FernauGHHVW19} the \emph{constrained synchronization problem} 
was defined for a fixed PDFA
$\mathcal B = (\Sigma, P, \mu, p_0, F)$. 

\begin{decproblem}\label{def:problem_L-constr_Sync}
  \problemtitle{\cite{FernauGHHVW19}~\textsc{$L(\mathcal B)$-Constr-Sync}}
  \probleminput{DCSA $\mathcal A = (\Sigma, Q, \delta)$.}
  \problemquestion{Is there a synchronizing word $w \in \Sigma^*$ for $\mathcal A$ with  $w \in L(\mathcal B)$?}
\end{decproblem}

The automaton $\mathcal B$ will be called the \emph{constraint automaton}.
If an automaton~$\mathcal A$ is a yes-instance of \textsc{$L(\mathcal B)$-Constr-Sync} we call $\mathcal A$ \emph{synchronizing with respect to~$\mathcal{B}$}. 
Occasionally,
we do not specify $\mathcal{B}$ and rather talk about \textsc{$L$-Constr-Sync}.

Previous
results have shown that unconstrained synchronization is solvable in
polynomial time, and constrained synchronization in polynomial space.

\begin{theorem}[\cite{Vol2008}] \label{thm:unrestricted_sync_poly_time}
	We can decide $\Sigma^*\textsc{-Constr-Sync}$
	in time $O(|\Sigma||Q|^2)$.
\end{theorem}


\begin{theorem}[\cite{FernauGHHVW19}]
  \label{thm:L-contr-sync-PSPACE}
  For any constraint automaton $\mathcal B = (\Sigma, P, \mu, p_0, F)$
  the problem \textsc{$L(\mathcal B)$-Constr-Sync} is in $\PSPACE$.
\end{theorem}

In~\cite{FernauGHHVW19}, a complete analysis of the complexity landscape when the constraint language is given by small partial automata was done. It is natural to extend this result to other language classes.

\begin{theorem}[\cite{FernauGHHVW19}]
\label{thm:classification_MFCS_paper}
 Let $\mathcal B = (\Sigma, P, \mu, p_0, F)$
 be a PDFA.
 If $|P|\le 1$ or $|P| = 2$ and $|\Sigma|\le 2$, then $L(\mathcal B)\textsc{-Constr-Sync} \in \PTIME$.
 For $|P| = 2$ with a ternary alphabet $\Sigma = \{a,b,c\}$, up to symmetry by renaming of the letters,
 $L(\mathcal B)\textsc{-Constr-Sync}$
 is $\PSPACE$-complete precisely in the following cases for $L(\mathcal B)$:
 $$
  \begin{array}{llll}
    a(b+c)^*        & (a+b+c)(a+b)^*  & (a+b)(a+c)^* & (a+b)^*c \\
    (a+b)^*ca^*     & (a+b)^*c(a+b)^* & (a+b)^*cc^*  & a^*b(a+c)^* \\
    a^*(b+c)(a+b)^* & a^*b(b+c)^*     & (a+b)^*c(b+c)^* & a^*(b+c)(b+c)^*
  \end{array}
 $$
 and polynomial time solvable in all other cases.
\end{theorem}

For $|P| = 3$ and $|\Sigma| = 2$,
the following is known: In~\cite{FernauGHHVW19} it has been shown
that $(ab^*a)\textsc{-Constr-Sync}$ is \NP-complete for general input automata. In~\cite[Theorem 33]{FernauGHHVW19} it was shown that 
$(b(aa+ba)^*)\textsc{-Constr-Sync}$
is \PSPACE-complete for general input automata. Further, it can be shown that for the following constraint languages the constrained problem
is \PSPACE-complete: $b^*a(a+ba)^*$, $a(b+ab)^* + b(bb^*a)^*$,
which are all acceptable by a $3$-state PDFA over a binary alphabet. So, with Theorem~\ref{thm:classification_MFCS_paper}, these are the smallest possible constraint automata over a binary alphabet giving \PSPACE-complete problems.

\todo{In der Tabelle bei simple idempotents tauscht zweimal P auf, nicht gut!}

\begin{table}[t]
    \centering
\begin{tabular}{llllllll} 
 Input Aut. Type       & Complexity Class & Hardness                        & Reference  \\ \hline
  General Automata          & \PSPACE          & \PSPACE-hard for $a(b+c)^*$ & \cite{FernauGHHVW19}  \\
 With Sink State           & \PSPACE          & \PSPACE-hard for $a(b+c)^*$ & \cite{FernauGHHVW19}  \\
 Weakly Acyclic            & \NP              & \NP-hard for $a(b+c)^*$     & \cite{DBLP:conf/dlt/Hoffmann21} \\ 
 Simple Idempotents        & \PTIME\ for $|\Sigma| = 2, |P| \le 3$                     & - & Theorem~\ref{thm:simple_idempotents_binary_in_P}  \\
 Commutative               & \PTIME           & - & Theorem~\ref{thm:main_theorem_comm_input}
\end{tabular}
    \caption{Overview of known result of the complexity landscape of $L(\mathcal B)\textsc{-Constr-Sync}$
    with $\mathcal B = (\Sigma, P, \mu, p_0, F)$ 
    when restricted to certain input automata.
    Constraint languages giving intractable problems are written next to the hardness claim. For a binary alphabet, the $3$-state PDFA language $b(aa+ba)^*$ gives an \PSPACE-complete problem for general input automata.}
    \label{fig:tab_overview_input_aut}
\end{table}

For an overview of the results for different classes of input automata, see Table~\ref{fig:tab_overview_input_aut}.

\section{Synchronizing Commutative Semi-Automata under Arbitrary Regular Constraints}
\label{sec:comm_input}

Here, we show that for commutative input automata
and an arbitrary regular constraint language,
the constrained synchronization problem is always solvable in polynomial time.

Subsection~\ref{subsec:structure_sync_words} is concerned with the set of synchronizing words for commutative semi-automata.
The main result of this subsection is Proposition~\ref{prop:com_waa_same_sync_words}, stating
that the set of synchronizing words for a commutative automaton
can be represented by a weakly acyclic and commutative automaton
computable in polynomial time.
Then, in Subsection~\ref{subsec:recognizing_intersection},
we show an auxiliary result that is also of independent
interest, namely that, given $m$ weakly acyclic
and commutative automata, the set of words accepted
by them all is recognizable by a weakly acyclic automaton
computable in polynomial time for a fixed alphabet.
Then, in Subsection~\ref{subsec:poly_alg},
we combine all these results
to show Theorem~\ref{thm:main_theorem_comm_input},
the statement that for commutative input automata
the constrained synchronization problem is tractable.

\subsection{The Structure of the Set of Synchronizing Words for Commutative Automata}
\label{subsec:structure_sync_words}

In commutative semi-automata, a synchronizing state must be a sink state, a property
not true for general semi-automata.

\begin{lemma}[\cite{FernauHoffmann19}] 
\label{lem:sync_state_is_sink}
 Let $\mathcal A = (\Sigma, Q, \delta)$
 be a commutative semi-automaton.
 Then, the synchronizing state must
 be a sink state (and hence is unique).
\end{lemma} 

The next lemma is important in
the proof of Proposition~\ref{prop:com_waa_same_sync_words}.

\begin{lemmarep} 
\label{lem:sync_comm_aut_SCC}
 Let $\mathcal A = (\Sigma, Q, \delta)$ be a commutative semi-automaton, $S \subseteq Q$ a strongly connected component and $u \in \Sigma^*$.
 Then the states in $\delta(S, u)$
 are pairwise connected.
\end{lemmarep}
\begin{proof}
 Let $t_1, t_2 \in \delta(S, u)$.
 Then $t_1 = \delta(s_1, u)$ and $t_2 = \delta(s_2, u)$
 for some $s_1, s_2 \in S$.
 As $S$ is a strongly connected component, we find $v,w \in \Sigma^*$
 such that $s_1 = \delta(s_2, v)$
 and $s_2 = \delta(s_1, w)$. 
 But then, by commutativity, 
 $t_1 = \delta(s_1, u) = \delta(\delta(s_2, v), u) = \delta(s_2, vu) = \delta(s_2, uv) = \delta(\delta(s_2, u),v) = \delta(t_2, v)$
 and similarly $t_2 = \delta(t_1, w)$.
\end{proof}

\begin{toappendix} 
\begin{corollaryrep}
\label{cor:two_SCC}
 Let $\mathcal A = (\Sigma, Q, \delta)$ be a commutative semi-automaton,
 $S,T \subseteq Q$ be two (not necessarily distinct) strongly connected components
 and $u \in \Sigma$. Then either $\delta(S, u) \cap T = \emptyset$
 or $\delta(S, u) \subseteq T$.
\end{corollaryrep}
\begin{proof}
 If $\delta(S, u) \cap T \ne \emptyset$,
 then, by Lemma~\ref{lem:sync_comm_aut_SCC}
 and the maximality of $T$, $\delta(S, u) \subseteq T$.
\end{proof}
\end{toappendix}

For commutative automata, the set of synchronizing
words is represented by a weakly acyclic automata that
is constructed out of the strongly connected components
and computable in polynomial time. 

\begin{proposition}
\label{prop:com_waa_same_sync_words}
 Let $\mathcal A = (\Sigma, Q, \delta)$
 be a synchronizing commutative semi-automaton
 with $n$ states.
 Then, there exists a weakly acyclic commutative
 semi-automaton with at most $n$ states and the same set of synchronizing
 words computable in polynomial time
 for a fixed alphabet.
\end{proposition}
\begin{proof}
 Define an equivalence relation on $Q$
 by setting two states $q, q' \in Q$ to be equivalent
 $q \sim q'$
 iff they are contained in the same strongly connected
 component. By Lemma~\ref{lem:sync_comm_aut_SCC}
 this is in fact a congruence relation, i.e.,
 if $q \sim q'$, then $\delta(q, u) \sim \delta(q', u)$
 for every $u \in \Sigma^*$.

 Let $q_1, \ldots, q_m \in Q$
 be a transversal of the equivalence classes, i.e.,
 we pick precisely one element from each class.
 Define $\mathcal C = (\Sigma, S, \mu)$
 with $S = \{q_1, \ldots, q_m\}$
 and $\mu(q_i, x) = q_j$
 iff $\delta(q_i, x) \sim q_j$.
 As we have a congruence relation, we get the same
 automaton for every choice of transversal $q_1, \ldots, q_m$
 (in fact, we only note in passing that 
 $\mathcal C$ is a homomorphic image of $\mathcal A$
 and corresponds to quotiening the automaton
 by the introduced congruence relation).
 Next, we argue that the sets of synchronizing
 words of both automata coincide.
 
 Let $s_f$ be the synchronizing state of $\mathcal A$.
 By Lemma~\ref{lem:sync_state_is_sink} it is a sink state,
 and so $\{s_f\}$ is a strongly connected component
 and we can deduce $s_f \in S$.
 Suppose $u \in \Sigma^*$ is a synchronizing
 word of $\mathcal A$.
 So, for each state $q_i$ from the transversal,
 we have $\delta(q_i, u) = s_f$,
 which implies, as, inductively, for every prefix $v$
 of $u$ we have $\delta(q_i, v) \sim \mu(q_i, v)$,
 that $\mu(q_i, u) = s_f$
 and $u$ synchronizes $\mathcal C$.
 Conversely, suppose $u \in \Sigma^*$
 synchronizes $\mathcal C$
 and let $q \in Q$.
 Then $q \sim q_i$ for some $q_i \in S$.
 Again, for every prefix $v$ of $u$,
 we have $\delta(q, v) \sim \mu(q_i, v)$,
 as is easy to see by the definition of $\mathcal C$
 and as $\sim$ is a congruence relation.
 So, $\delta(q, u) \sim s_f$, which implies,
 as $\{s_f\}$ is a strongly connected component,
 that $\delta(q,u) = s_f$.
\end{proof}

\begin{remark} 
 See Figure~\ref{fig:ex_comm_aut}
 for a synchronizing commutative automaton 
 and a weakly acyclic automaton constructed as in the proof
 of Proposition~\ref{prop:com_waa_same_sync_words}
 having the same set of synchronizing words.
 Note that the construction can actually be performed
 for any commutative automaton, even a non-synchronizing
 one as shown in Figure~\ref{fig:ex_comm_aut2}.

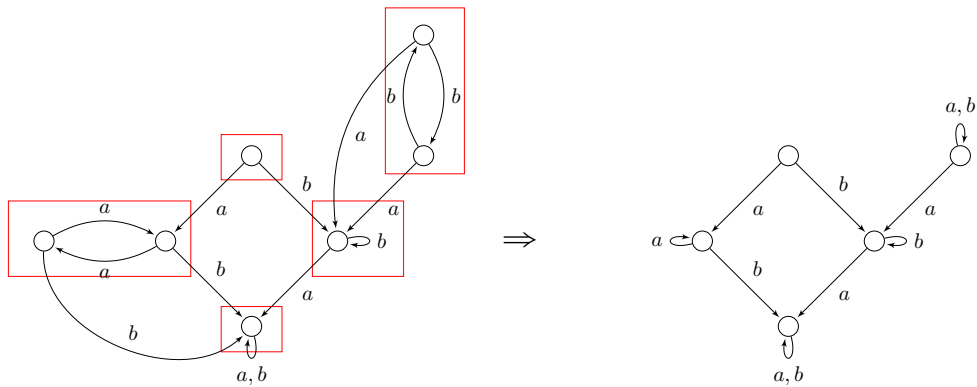
\begin{figure}[htb]
     \centering
    \scalebox{.8}{
   
  \begin{tikzpicture}[>=latex',shorten >=1pt,node distance=2cm,on grid,auto]
 \tikzset{every state/.style={minimum size=1pt}}
 
   \node[state] (2) {};
   \node[state] (3) [below left of=2] {};
   \node[state] (4) [below right of=2] {};
   \node[state] (5) [left of=3] {};
   \node[state] (6) [below right of=3] {};
   \node[state] (8) [above right of=4] {};
   \node[state] (9) [above of=8] {};
   
   \path[->] 
             (2) edge node {$a$}   (3)
             (3) edge [bend left] node  {$a$}   (5)
             (5) edge [bend left]  node {$a$}   (3)
             (3) edge node {$b$} (6)
             (6) edge [loop below] node {$a,b$} (6)
             (5) edge [bend right=70]  node  {$b$} (6);
             
   \path[->] (2) edge node {$b$} (4)
             (4) edge [loop right] node {$b$} (4)
             (4) edge node {$a$} (6);
             
   \path[->] (8) edge [bend left] node {$b$} (9)
             (9) edge [bend left] node {$b$} (8)
             (8) edge node {$a$} (4)
             (9) edge [bend right] node {$a$} (4);

   
   \node (P)  [right = 3cm of 4] {\LARGE $\Rightarrow$};
   \node[state] (S1) [right = 3cm of P] {};
   \node[state] (S2) [above right of=S1] {};
   \node[state] (S3) [below right of=S2] {};
   \node[state] (S4) [below left  of=S3] {};
   \node[state] (S5) [above right of=S3] {};
   
   \path[->] (S2) edge node {$a$} (S1)
             (S2) edge node {$b$} (S3)
             (S3) edge node {$a$} (S4)
             (S1) edge node {$b$} (S4)
             (S5) edge node {$a$} (S3);
   
   \path[->] (S1) edge [loop left] node {$a$} (S1)
             (S3) edge [loop right] node {$b$} (S3)
             (S5) edge [loop above] node {$a,b$} (S5)
             (S4) edge [loop below] node {$a,b$} (S4);
             
        \begin{pgfonlayer}{background}
                \draw [draw=red] (-4,-0.75) rectangle (-1,-2);
                \draw [draw=red] (0.5,-2.5) rectangle (-0.5,-3.25);
                \draw [draw=red] (0.5,-0.4) rectangle (-0.5,0.35);
                \draw [draw=red] (1,-0.75) rectangle (2.5,-2);
                \draw [draw=red] (2.2,-0.3) rectangle (3.5,2.45);
        \end{pgfonlayer}
                
\end{tikzpicture}
}
  \caption{A synchronizing commutative automaton
    and the weakly acyclic automaton from the proof
    of Proposition~\ref{prop:com_waa_same_sync_words}
    with the same set of synchronizing words.
    A shortest synchronizing word is $baa$. The strongly connected components
    in the original automaton have been framed by red boxes.}
  \label{fig:ex_comm_aut}
\end{figure}

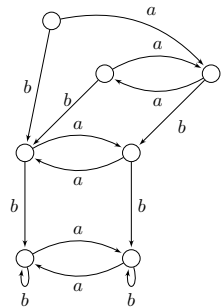
\begin{figure}[htb]
     \centering
    \scalebox{.7}{
   
  \begin{tikzpicture}[>=latex',shorten >=1pt,node distance=2cm,on grid,auto]
 \tikzset{every state/.style={minimum size=1pt}}
    \node[state] (10) at (11,2) {};
   \node[state] (11) at (14,1) {}; 
   \node[state] (12) at (10.5,-0.5) {}; 
   \node[state] (13) [below of=12] {};
   \node[state] (14) [right of=13] {};
   \node[state] (15) [left of=11] {};
   \node[state] (16) [right of=12] {};
   
   \path[->] (10) edge [bend left] node {$a$} (11)
             (11) edge node {$b$} (16)
             (10) edge [left] node {$b$} (12)
             (12) edge [left] node {$b$} (13)
             (13) edge [loop below] node {$b$} (12)
             (14) edge [loop below] node {$b$} (14)
             (13) edge [bend left] node {$a$} (14)
             (14) edge [bend left] node {$a$} (13);
    
 
   \path[->] (11) edge [bend left] node {$a$} (15)
             (15) edge [bend left] node {$a$} (11);
             
   \path[->] (16) edge [bend left] node {$a$} (12)
             (12) edge [bend left] node {$a$} (16);
             
   \path[->] (15) edge [left,pos=.3] node {$b$} (12);
   \path[->] (16) edge node {$b$} (14);
   
\end{tikzpicture}
}
   
   \caption{A commutative automaton that is not synchronizing.}
  \label{fig:ex_comm_aut2}
\end{figure}

\end{remark}

With the method of proof from Proposition~\ref{prop:com_waa_same_sync_words}, we can slightly improve the running time stated in Theorem~\ref{thm:unrestricted_sync_poly_time}
for commutative input semi-automata.

\begin{corollaryrep} 
\label{cor:sync_comm_aut_decision_problem}
 Let $\mathcal A = (\Sigma, Q, \delta)$ be a commutative semi-automaton.
 Then it can be decided in time 
 $O(|Q| + |\Sigma||Q|)$
 if $\mathcal A$ is synchronizing.
\end{corollaryrep}
\begin{proof}
 For a given directed acyclic graph with $n$ nodes and $m$ edges,
 a topological sorting\todo{nochmal definieren was das ist?}
 can be performed in time $O(n + m)$, see~\cite{DBLP:books/daglib/0023376}.
 Hence, if we consider the directed acyclic graph resulting from 
 the strongly connected components~\cite{DBLP:books/daglib/0023376}
 of the automaton graph of $\mathcal A$,
 we can topologically sort it in time $O(|Q| + |\Sigma||Q|)$,
 as we have at most $|Q|$ strongly connected components
 and $|\Sigma||Q|$ many edges among them, as this is an upper bound for the edges $\mathcal A$,
 as each state contributes $|\Sigma|$ many edges.
 The strongly connected components themselves can also be computed
 in time $O(|Q| + |\Sigma||Q|)$ with a similar argument, see~\cite{DBLP:books/daglib/0023376}.

 \todo{das vielleicht nur in proof sketch bringen.}
 Another argument computes
 first the strongly connected components,
 from which the automaton $\mathcal C$
 from the proof of Proposition~\ref{prop:com_waa_same_sync_words}
 can be derived, and then uses the fact~\cite[Corollary 6]{DBLP:conf/dlt/Hoffmann21a} that
 for weakly acyclic automata
 we can decide synchronizability
 in $O(m + |\Sigma|m)$, where $m$
 denotes the number of strongly connected components
 of $\mathcal A$.
\end{proof}

Let us note the following consequence
of Lemma~\ref{lem:aut_for_sync_words}
and the bound $n - 1$
for a shortest synchronizing word in weakly acyclic
automata~\cite[Proposition 1]{DBLP:journals/tcs/Ryzhikov19a}
and a bound for the shortest synchronizing word
with respect to a constraint~\cite[Proposition 7]{DBLP:conf/dlt/Hoffmann21a}.
This gives an alternative proof of the tight bound $n - 1$
for commutative automata from~\cite{FernauHoffmann19,Rys96}.

\begin{corollary}
 If $\mathcal A = (\Sigma, Q, \delta)$ is a synchronizing commutative semi-automaton
 with $n$ states,
 then there exist a synchronizing word of length
 at most $n - 1$ and, for any constraint PDFA $\mathcal B$,
 there exists a synchronizing word in $L(\mathcal B)$
 of length at most $|P| \binom{n}{2}$.
\end{corollary}
\begin{proof}
 By Lemma~\ref{lem:aut_for_sync_words}
 the set of synchronizing words
 equals the set of synchronizing words
 of a weakly acyclic automaton, for which we have the known
 bounds.
\end{proof}

In fact, with a little more work and generalizing  Lemma~\ref{lem:aut_for_sync_words}, using~\cite[Proposition 1]{DBLP:journals/tcs/Ryzhikov19a},
we can show the stronger statement that for commutative $\mathcal A$, if there
exists a word $w \in \Sigma^*$
with $|\delta(Q, w)| = r$, then
there exists one of length at most $n - r$.

\subsection{Recognizing the Intersection of Weakly Acyclic Commutative Automata}
\label{subsec:recognizing_intersection}

\begin{proposition}
\label{prop:intersection_waa_com}
 Let $\mathcal A_i = (\Sigma, Q, \delta, q_i, F_i)$,
 $i \in \{1,\ldots,m\}$,
 be weakly acyclic and commutative automata
 with at most $n$ states.
 Then, $\bigcap_{i=1}^m L(\mathcal A_i)$
 is recognizable by a weakly acyclic
 commutative automaton 
 of size $n^{|\Sigma|}$
 computable in polynomial-time
 for a fixed alphabet.
\end{proposition}
\begin{proof}
 Let $\Sigma = \{a_1, \ldots, a_k\}$.
 Define the threshold counting function $\psi_n \colon \Sigma^* \to \{0,1,\ldots,n-1\}^k$
 by 
 $
 \psi_n(u) = (\min\{n-1,|u|_{a_1}\}, \ldots, \min\{n-1,|u|_{a_k}\}).
 $
 
 Let $i \in \{1,\ldots,m\}$. Suppose $u \in \Sigma^*$.
 If $|u|_{a_j} \ge n$
 for some $j \in \{1,\ldots,k\}$,
 then $\mathcal A_i$
 must traverse at least one self-loop labeled
 by $a_j$ 
 when reading $u$, as $\mathcal A_i$
 is weakly acyclic.
 So, by not traversing these self-loops an appropriate
 number of times,
 we find a word $u' \in \Sigma^*$
 with $|u'|_{a_j} = \min\{n-1, |u|_{a_j}\}$,
 which implies $\psi_n(u) = \psi_n(u')$,
 and $\delta_i(q_i, u) = \delta_i(q_i, u')$.

 
 For $a \in \Sigma$, set $a^{\le n-1} = \{\varepsilon, a, \ldots, a^{n-1} \}$.
 Define $E_i = \{ \psi_n(u) \mid u \in a_1^{\le n-1} \cdots a_k^{\le n-1} \cap L(\mathcal A_i) \}$.
 Let $u \in L(\mathcal A_i)$,
 by the above and using that $\mathcal A_i$
 is commutative, 
 we can deduce $\psi_n(u) \in E_i$.
 If $u \in \psi_n^{-1}(E_i)$,
 then there exists $v = a_1^{c_1} \cdots a_k^{c_k}$
 with $c_j \in \{0,1,\ldots,n-1\}$
 for all $j \in \{1,\ldots,k\}$
 and $\psi_n(u) = \psi_n(v)$.
 As $\psi_n(v) \in E_i$,
 we have $v \in L(\mathcal A_i)$.
 As before, there exists $u' \in \Sigma^*$
 with $|u'|_{a_j} = \min\{ n-1, |u|_{a_j} \}$
 for all $j \in \{1,\ldots,k\}$
 such that $\delta_i(q_i, u) = \delta_i(q_i, u')$
 and $\psi_n(u) = \psi_n(u')$.
 Hence, $\psi_n(u') = \psi_n(v)$,
 and as $|u|_{a_j} < n$
 for all $j \in \{1,\ldots,k\}$
 and by commutativity,
 this implies $u' \in L(\mathcal A_i)$,
 which, furthermore,
 implies $u \in L(\mathcal A_i)$,
 as both end up in the same state of $\mathcal A_i$.
 Summarizing, we have shown
 \[
  u \in L(\mathcal A_i) 
  \Leftrightarrow u \in \psi_n^{-1}(E_i).
 \]
 As $i \in \{1,\ldots,m\}$ was chosen arbitrarily,  $\bigcap_{i=1}^m L(\mathcal A_i)
 = \bigcap_{i=1}^m \psi_n^{-1}(E_i)
 = \psi_n^{-1}(\bigcap_{i=1}^m E_i)$.

 Finally, a language of the form $\psi_n^{-1}(E)$
 is recognizable by a weakly acyclic commutative
 automaton: set $\mathcal A = (\Sigma, Q, \delta, q_0, E)$
 with $Q = \{0, 1, \ldots, n-1\}$
 and $\delta((s_1, \ldots, s_k), a_j)
 = (s_1, \ldots, s_{j-1}, \min\{ n - 1, s_j + 1 \}, s_{j+1}, \ldots, s_k)$ and $q_0 = (0,\ldots,0)$. It is obvious
 that $\mathcal A$ is commutative and $L(\mathcal A) = \psi_n^{-1}(E)$ as $\mathcal A$ essentially
 implements the threshold counting expressed by $\psi_n$.
 
 As $a_1^{\le n - 1} \cdots a_k^{\le n - 1}$
 contains $n^k$ words, the sets $E_i$
 and their intersection
 can be computed in polynomial time.
 Also, the automaton recognizing
 $\psi_n^{-1}(\bigcap_{i=1}^m E_i)$
 is computable in polynomial time.
\end{proof}

\begin{remark}
 With Proposition~\ref{prop:intersection_waa_com},
 we can deduce that we can decide in 
 time $O(n^{f(|\Sigma|)})$,
 for some computable function $f$,
 if the intersection of the languages
 recognized by $m$
 weakly acyclic commutative automata
 is non-empty. Hence, this
 problem is in \XP\  for these types
 of automata and parameterized 
 by the size of the alphabet (see~\cite{FluGro2006} for an introduction
 to parameterized complexity theory).
 It can be shown that the languages
 recognized by these types of automata
 recognize star-free languages~\cite{Arrighi2021}.
 Note that in~\cite{Arrighi2021},
 the mentioned conclusion has been
 improved to an \XP\ result for so called
 totally star-free non-deterministic automata
 recognizing commutative languages.
 We refer to the mentioned paper for details.
 Contrary, for general commutative, even unary,
 automata
 this problem is $\NP$-complete~\cite{FernauHoffmannWehar2021},
 and, further, for fixed alphabet sizes,
 $W[1]$-complete with the number of input
 automata as parameter, see~\cite{FernauHoffmannWehar2021}.
\end{remark}

\subsection{A Polynomial-Time Algorithm for the Constrained Synchronization Problem with Commutative Input Semi-Automata}
\label{subsec:poly_alg}

\begin{lemma}
\label{lem:aut_for_sync_words}
 Let $\mathcal A = (\Sigma, Q, \delta)$
 be a commutative semi-automaton with $n$ states.
 Then, the set of synchronizing words
 is recognizable by an automaton
 of size $n^{|\Sigma|}$
 computable in polynomial time 
 for a fixed alphabet.
\end{lemma}
\begin{proof}

 First, we can test in polynomial time if $\mathcal A$
 is synchronizing. If not, an automaton
 with a single state and empty set of final states
 recognizes the empty set.
 Otherwise, by Proposition~\ref{prop:com_waa_same_sync_words}
 we can compute a weakly acyclic automaton $\mathcal C = (\Sigma, S, \mu)$
 having the same set of synchronizing
 words in polynomial time.
 By Lemma~\ref{lem:sync_state_is_sink}, there exists a unique
 synchronizing state $s_f \in S$. \todo{diese notation nutze ich acuh später, am anfang einführen?}
 For $q \in S$, let $\mathcal C_{q,\{s_f\}} = (\Sigma, S, \mu, q, \{s_f\})$ be the automaton $\mathcal C$ but with 
 start state $q$ and set of final states $\{s_f\}$. 
 Then, the set of synchronizing words of 
 $\mathcal C$ is
 \[
  \bigcap_{q \in Q} L(\mathcal C_{q, \{s_f\}}).
 \]
 By Proposition~\ref{prop:intersection_waa_com}, we
 can compute in polynomial time
 a weakly acyclic commtuative automaton of size $n^{|\Sigma|}$
 recognizing
 the above set.
\end{proof}

\begin{toappendix}
\begin{remark}
 In the proof of Lemma~\ref{lem:aut_for_sync_words},
 as $\mathcal C = (\Sigma, S, \mu)$ is weakly acyclic,
 we can show that
 \[
  \bigcap_{q \in Q} L(\mathcal C_{q, \{s_f\}})
   = \bigcap_{\substack{q \in Q \\ \forall p \in Q \setminus\{q\} \  \forall u \in \Sigma^* : \mu(p, u) \ne q}} L(\mathcal C_{q, \{s_f\}}),
 \]
 which improves the running time of the algorithm.
 In terms of the partial order given on the states
 by the reachability relation,\todo{diese irgendwo definieren}
 this is implied as every letter either moves to 
 a larger state or induces a self-loop,
 which implies that we can take the above
 intersection over the minimal states with respect
 to the partial order.
\end{remark}
\end{toappendix}


Now, combining the above results we can prove
our main result concerning the constrained
synchronization problem.

\begin{theorem}
\label{thm:main_theorem_comm_input}
 Let $L \subseteq \Sigma^*$ be regular.
 Then, for a commutative input semi-automaton $\mathcal A$
 with $n$ states,
 the problem if $\mathcal A$ admits a synchronizing
 word in $L$ is solvable in polynomial time.
\end{theorem}
\begin{proof}
 By Lemma~\ref{lem:aut_for_sync_words}, we
 can compute in polynomial time an
 automaton $\mathcal C$ of size $n^{|\Sigma|}$
 recognizing the set
 of synchronizing words of $\mathcal A$.
 By the product automaton construction~\cite{HopUll79},
 using an automaton for $L$,
 we can compute in polynomial
 time a recognizing automaton
 for the intersection $L(\mathcal C) \cap L$.
 Then, checking that the resulting automaton
 for the intersection recognizes a non-empty
 language can also be done in polynomial
 time~\cite{HopUll79}.
\end{proof}

The degree of the polynomial measuring the running time depends on the alphabet size. But note that a fixed
constraint automaton also fixes the alphabet, hence this parameter is not allowed to vary in the input semi-automata.

\begin{toappendix} 
\todo{noch richtig einordnen in appendix}
In this section, we take a closer look at the structure
of commutative semi-automata which admit a synchronizing word.
We show that the connectivity of a strongly connected component
is preserved under the action of a letter, that a letter
either permutes the states of some component, or maps every state
into another component, and lastly that a letter, once it permutes
a strongly connected component, it must also permute every component reachable
from the one under consideration. Intuitively, this means a letter
either moves forward between the components, or it has to stop. Then, we use these properties to give a simple criterion
for general synchronizability of a commutative semi-automaton.
\end{toappendix}

\begin{toappendix}
\todo{Ist das nicht vorheriges Corollar mit $T = S$?}

\begin{lemmarep}
\label{lem:maps_outside_or_permutes}
 Let $\mathcal A = (\Sigma, Q, \delta)$ be a commutative semi-automaton
 and $S \subseteq Q$ be a strongly connected component. 
 Then, each letter $x \in \Sigma$ either permutes all states in $S$,
 or maps every state in $S$ to some state outside of $S$, i.e.,
 we either have $\delta(S, x) = S$ or $\delta(S, x) \cap S = \emptyset$.
\end{lemmarep}
\begin{proof}
 Suppose $\delta(S, x) \cap S \ne \emptyset$.
 Let $t \in S$. As $S$ is a strongly connected component,
 we find $u \in \Sigma^*$ such that $t \in \delta(\delta(S, x) \cap S, u) \subseteq \delta(S, u)$.
 Hence, $\delta(S, u) \cap S \ne \emptyset$.
 So, by Corollary~\ref{cor:two_SCC}, $\delta(S,u) \subseteq S$.
 By commutativity, $t \in \delta(S, xu) = \delta(S, ux) = \delta(\delta(S,u),x) \subseteq \delta(S,x)$.
 Hence, $S \subseteq \delta(S, x)$.
 Lastly, we can either note, as $S$ is finite and $|\delta(S,x)| \le |S|$, that
 the previous inclusion implies $S = \delta(S, x)$, or we can reason as before, i.e., using 
 Corollary~\ref{cor:two_SCC} to deduce $\delta(S, x) \subseteq S$. Hence $S = \delta(S, x)$, 
 i.e., $x$ permutes the states in~$S$.
\end{proof}


\begin{lemmarep}
\label{lem:whole_words}
 Let $\mathcal A = (\Sigma, Q, \delta)$ be a commutative semi-automaton, $S \subseteq Q$ be a strongly connected component and
 $u,v \in \Sigma^*$.
 If $\delta(S, uv) \subseteq  \delta(S,u)$,
 then $\delta(S,uv) = \delta(S,u)$.
\end{lemmarep}
\begin{proof}
 By Lemma~\ref{lem:sync_comm_aut_SCC}, we find a strongly connected component $T\subseteq Q$ containing
 $\delta(S,u)$. Then, with the assumption,
 $\delta(T,v)\cap T \ne \emptyset$.
 But this implies $\delta(T, v) = T$,
 as by Lemma~\ref{lem:maps_outside_or_permutes},
 if we write $v = x_1 \cdots x_n$ with $x_i \in\Sigma$, $i \in \{1,\ldots,n\}$,
 the first letter $x_1$ has to permute $T$, for if it does not,
 it will map inside another component. But then, 
 as the strongly connected components form an acyclic graph,
 it can never come back to $T$.
 So, $\delta(T,x_1 x_2 \dots x_n)= \delta(T,x_2\cdots x_n)$
 and with a similar argument, $x_2$
 has to permute $T$ and so on. Hence $v$
 is a concatenation of permutations on $T$,
 hence permutes $T$ itself.
 But then, $v$ acts injective
 on the subset $\delta(S,u) \subseteq T$,
 and so $\delta(S, uv) \subseteq  \delta(S,u)$
 implies $\delta(S,uv) = \delta(S,u)$.
\end{proof}

\begin{lemmarep}
\label{lem:letter_maps_further_or_stucks_for_rest}
 Let $\mathcal A = (\Sigma, Q, \delta)$ be a commutative semi-automaton
 and $S, T \subseteq Q$ be two strongly connected components.
 Let $x \in \Sigma$ and suppose we have $u \in (\Sigma \setminus \{x\})^*$
 such that $\delta(S, x) = S$ and $\delta(S, u) \subseteq T$. 
 Then $\delta(T, x) = T$. 
\end{lemmarep}
\begin{proof} 
 We have $\delta(S, u) = \delta(S, xu) = \delta(S, ux) = \delta(\delta(S, u), x) \subseteq \delta(T, x)$.
 So, $T \cap \delta(T, x) \ne \emptyset$,
 which implies, by Lemma~\ref{lem:maps_outside_or_permutes},
 that $\delta(T, x) = T$.
\end{proof}

By Lemma~\ref{lem:sync_comm_aut_SCC}, 
Lemma~\ref{lem:maps_outside_or_permutes} and Lemma~\ref{lem:letter_maps_further_or_stucks_for_rest},
we already know much about the structure of commutative semi-automata.
We know that any letter either maps completely outside of any strongly connected component
or permutes the states within this component,
and if it maps outside, all target states lie within a common strongly connected
component. 
Also, we know, if we topologically sort the strongly connected
components, that along each path, either a letter always maps everything 
to some strongly connected component strictly greater in the linear ordering
given by the topological sorting, or it permutes the states in this components, but
then it will also do in every component reachable from the one under consideration.
Examples of commutative semi-automata
are presented in Figure~\ref{fig:ex_comm_aut}.
\end{toappendix}

\begin{toappendix}

%
%
Next, we give an easy criterion when a commutative semi-automaton
is synchronizing.


\begin{proposition}
\label{prop:if_sync_form_SCCs}
 Let $\mathcal A = (\Sigma, Q, \delta)$ be a commutative semi-automaton
 and $S_1, \ldots, S_n \subseteq Q$
 be any topological sorting of the strongly connected
 components of $\mathcal A$.
 Then $\mathcal A$ is synchronizing if and only if
 $S_n$ is a singleton set and reachable from every other component.
\end{proposition}
\begin{proof}
 Assume $\mathcal A$ is synchronizing with synchronizing state $s_f$.
 As $s_f$ is reachable from every state and a sink state by
 Lemma~\ref{lem:sync_state_is_sink}, we must have $S_{n} = \{s_f\}$
 and $S_n$ is reachable from every other component.
 Conversely, assume $S_n = \{s_f\}$ and we have words $u_1, \ldots, u_{n-1}$
 such that $\delta(S_i, u_i) = \{s_f\}$ for $i \in \{1,\ldots, n-1\}$.
 As $S_n$ is the last component in a topological sorting, $\delta(S_n, x) \subseteq S_n$
 for any $x \in \Sigma$. So, $s_f$ is a sink state.
 Set $u = u_1 \cdots u_{n-1}$.
 Then, for any $S_i$, $i \in \{1,\ldots,n-1\}$,
 we have, by commutativity and as $s_f$ is a sink state, 
 \begin{align*}
   \delta(S_i, u) & = \delta(S_i,u_1 \cdots u_{n-1}) \\
                  & = \delta(S_i, u_i u_1 \cdots u_{i-1}u_{i+1}\cdots u_{n-1}) \\
                  & = \delta(\delta(S_i,u_i),u_1 \cdots u_{i-1}u_{i+1}\cdots u_{n-1}) \\
                  & = \delta(\{s_f\}, u_1 \cdots u_{i-1}u_{i+1}\cdots u_{n-1}) \\
                  & = \{s_f\}.
 \end{align*} 
 So $u$ is a synchronizing word and $s_f$ a synchronizing state.
\end{proof}

\begin{remark}
 The structure results derived here generalize the structure results
 for the extremal commutative semi-automata considered in~\cite{FernauHoffmann19}, i.e.,
 those for which a shortest synchronizing word has maximal length.
     
     
\end{remark}
\end{toappendix}

\section{Synchronizing Automata with Simple Idempotents over a Binary Alphabet}
\label{sec:simpl_idemp}

Here, we show that for input automata with simple idempotents over a binary alphabet
and a constraint given by a PDFA with at most three states, the constrained synchronization problem
is always solvable in polynomial time. Note that, as written at the end of Section~\ref{sec:preliminaries}, the smallest constraint languages
giving \PSPACE-complete problems are given by $3$-state automata
over a binary alphabet. But, as shown here, if we only allow automata with simple idempotents as input, the problem remains tractable in these cases.

Intuitively, by applying an idempotent letter, we can map at most two states to a single state. Hence,
to synchronize an $n$-state automaton with simple idempotents, we have to apply at least $n - 1$
times an idempotent letter. This is the content of the next lemma.

\begin{lemmarep}
\label{lem:each_sync_word_contains_a}
 Let $\mathcal A = (\Sigma, Q, \delta)$ be a semi-automaton with $n > 0$ states
 with simple idempotent and let $\Gamma \subseteq \Sigma$
 be the set of all idempotent letters that are not permutational letters.
 Suppose $w \in \Sigma^*$ is a synchronizing word for $\mathcal A$.
 Then, $\sum_{a \in \Gamma} |w|_a \ge n - 1$.
\end{lemmarep}
\begin{proof}
 As every letter $a \in \Gamma$ maps at most two states to a single state
 in each application, 
 we have $|S| - 1 \le |\delta(S, a)| \le |S|$
 for each $S \subseteq Q$.
 All other letters permute the states, hence the cardinality
 of each subset is invariant for them, i.e., $a \notin \Gamma$
 implies $|S| = |\delta(S, a)|$ for $S \subseteq Q$.
 This yields, inductively, $|S| - \sum_{a \in \Gamma} |w|_a \le |\delta(S, w)|$
 for $w \in \Sigma^*$. Hence, if $|\delta(Q, w)| = 1$
 we must have $\sum_{a \in \Gamma} |w|_a \ge n - 1$.
\end{proof}

The following was shown in~\cite[Proposition 6.2]{Mar2009b}. 

\begin{toappendix}
 The reader might notice, upon comparing~\cite[Proposition 6.2]{Mar2009b}
 with Proposition~\ref{prop:structure_simp_idem_binary},
 that the parameter $p$ is actually different. Here, it denotes the number of times
 we have to apply $b$ to reach the other state, in~\cite[Proposition 6.2]{Mar2009b},
 by $p$ a state is denoted for which this relation does not hold true.
 The change was made because there seems to be a (minor, the result still holds true of course)
 glitch in~\cite[Proposition 6.2]{Mar2009b}.
 For example, the author writes for ``[...] For $q_2 = n$, the automaton $\mathcal A$ is the \v{C}ern\'y
 automaton [...]'', but $\gcd(n,n) = n$, which would imply, according to the following conclusion in~\cite[Proposition 6.2]{Mar2009b},
 that then the automaton is not synchronizing for $n > 1$, which is not the case.
 Or, to give a more concrete example, an automaton with four states
 $\{1,2,3,4\}$ and the naming as in~\cite[Proposition 6.2]{Mar2009b}
 and $q_2 = 3$ is not synchronizing, but $4$ and $3$ are coprime.
\end{toappendix}

\begin{proposition}
\label{prop:structure_simp_idem_binary}
 Let $\Sigma = \{a,b\}$ be a binary alphabet
 and $\mathcal A = (\Sigma, Q, \delta)$
 be an $n$-state automaton with simple idempotents.
 Suppose $\mathcal A$ is synchronizing and $n > 3$.
 Then, up to renaming of the letters,
 we have only two cases for $\mathcal A$:
 
 \begin{enumerate} 
 \item There exists a sink state $t \in Q$,
  the letter $b$ permutes the states in $Q \setminus \{t\}$
  in a single cycle 
  and $|\delta(Q, a)| = n - 1$ with $t \in \delta(Q\setminus\{t\}, a)$.
 \item The letter $b$ permuates the states in $Q$ in a single cycle and there exists $0 < p < n$
  coprime to $n$ and $s, t\in Q$ such that $t = \delta(s, a) = \delta(s, b^p)$.
 \end{enumerate}
\end{proposition}

\begin{figure}[htb]
\centering
\begin{minipage}[t]{0.5\textwidth}
\scalebox{1.0}{\begin{tikzpicture}[shorten >=1pt,->,node distance=1.5cm]
  \tikzstyle{vertex}=[circle,draw,minimum size=16pt,inner sep=1pt]
 
  \foreach \name/\angle/\text in {A1/0/1,A2/45/2,A3/90/3,A4/135/4,A5/180/5,A6/225/6,A8/315/8}
  {
    \node[vertex] (\name) at (\angle:2cm) {};  
  }
  \node[vertex,right of=A1] (dangle) {$t$};
  
  \draw (A2) edge[loop above, above] node {$a$} (A2); 
  \draw (A3) edge[loop above, above] node {$a$} (A3); 
  \draw (A4) edge[loop above, above] node {$a$} (A4); 
  \draw (A5) edge[loop left] node {$a$} (A5); 
  \draw (A6) edge[loop left] node {$a$} (A6); 
  \path[->] (A8) edge[loop right] node {$a$} (A8); 
  
  \path[->] (A6) edge [left]  node {$b$} (A5)
            (A5) edge [left]  node {$b$} (A4)
            (A4) edge [above] node {$b$} (A3)
            (A3) edge [above] node {$b$} (A2)
            (A2) edge [right] node {$b$} (A1)
            (A1) edge [right] node {$b$} (A8);
 
 \draw (A1) edge [above] node {$a$} (dangle);
 \draw (dangle) edge [loop above] node {$a,b$} (dangle);
 
 \node at (0,-1.5) {\LARGE $\ldots$};
\end{tikzpicture}}%
\end{minipage}\hfill
\begin{minipage}[t]{0.45\textwidth}
\scalebox{1.0}{\begin{tikzpicture}[shorten >=1pt,->,node distance=1.5cm]
  \tikzstyle{vertex}=[circle,draw,minimum size=16pt,inner sep=1pt]
 
  \foreach \name/\angle/\text in {A1/0/,A2/45/,A3/90/,A4/135/,A5/180/t,A6/225/,A7/270/,A8/315/s}
  {
    \node[vertex] (\name) at (\angle:2cm) {\small $\text$};
  }
  
  \draw (A1) edge[loop right] node {$a$} (A1); 
  \draw (A2) edge[loop right] node {$a$} (A2); 
  \draw (A3) edge[loop below, below] node {$a$} (A3); 
  \draw (A4) edge[loop left] node {$a$} (A4); 
  \draw (A5) edge[loop left]  node {$a$} (A5); 
  \draw (A6) edge[loop left]  node {$a$} (A6); 
  \draw (A7) edge[loop above]  node {$a$} (A7); 
  \draw (A8) edge [above] node {$a$} (A5);
  
  \draw (A1) edge [right] node {$b$} (A8);
  \draw (A7) edge [above] node {$b$} (A6);
  \draw (A6) edge [left] node {$b$} (A5);
  \draw (A5) edge [left] node {$b$} (A4);
  \draw (A3) edge [above] node {$b$} (A2);
  \draw (A2) edge [right,pos=.4] node {$b$} (A1);
  
  \node [rotate=25] at (0.7,-1.7) {\LARGE $\ldots$};
  \node [rotate=25] at (-0.7,1.7) {\LARGE $\ldots$};
\end{tikzpicture}}%
\end{minipage}
 \caption{The two cases from Proposition~\ref{prop:structure_simp_idem_binary}. In the second case,
 for $p = 1$ with the notation from the statement, we get the automata from 
 the \v{C}erny family~\cite{Cerny64}, a family of automata giving the lower bound $(n-1)^2$ for the length of a shortest
 synchronizing word.}
 \label{fig:theorem}
\end{figure}
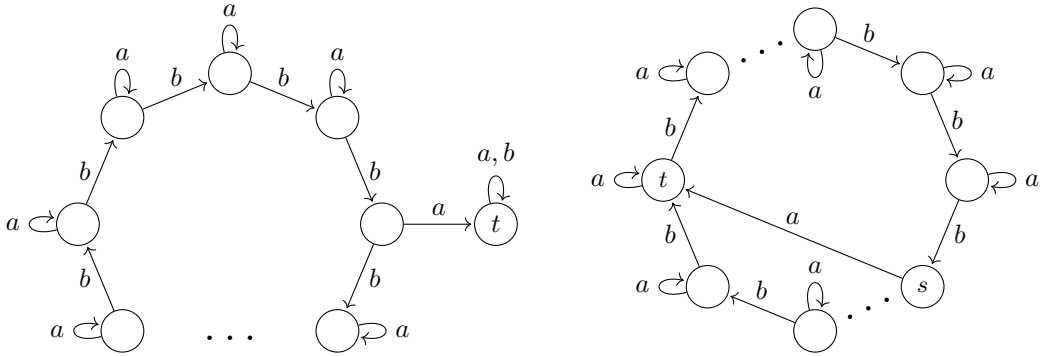

Note that the set of synchronizing words can be rather complicated in both cases.\todo{zeigen, exponential state complexity.} For example, the language
$
  \Sigma^*a^+(b((ba^*)^{n-1})^*)a^+)^{n-2}\Sigma^*
$
contains only synchronizing words for automata of the first type in Proposition~\ref{prop:structure_simp_idem_binary},
and, similarly,\todo{prüfen!}
$
 \Sigma^*a^+((ba^*)^{n-p}(ba^*)^n a^+)^{n-2}\Sigma^*
$
in the second case. However, for example, \todo{restklassen-zustaende syncen und bewegen argumetn.}
the language $(bbba)^*bb(bbba)^*bb(bbba)^*$ contains
synchronizing words and non-synchronizing words for automata of the first type.

Finally, for binary automata with simple idempotents
and an at most $3$-state constraint PDFA, the constrained
synchronization problem is always solvable in polynomial time.
The proof works by case analysis on the possible sequences of words in the constraint language and if it is possible to synchronize the two automata
types listed in Proposition~\ref{prop:structure_simp_idem_binary} with those sequences.

\begin{theoremrep}
\label{thm:simple_idempotents_binary_in_P}
 Let $\mathcal B = (\Sigma, P, \mu, p_0, F)$
 be a constraint automaton with $|P| \le 3$ and $|\Sigma| \le 2$.
 Let $\mathcal A = (\Sigma, Q, \delta)$
 be an input semi-automaton with simple idempotents. 
 Then, deciding if $\mathcal A$ has a synchronizing word in $L(\mathcal B)$
 can be done in polynomial time.
\end{theoremrep}
\begin{proofsketch}
 The cases $|\Sigma| \le 1$
 or $|P| \le 2$ and $|\Sigma| = 2$
 are polynomial time decidable in general
 as shown in~\cite[Corollary 9 \& Theorem 24]{FernauGHHVW19}.
 So, we can suppose $\Sigma = \{a,b\}$
 and $P = \{1,2,3\}$ with $p_0 = 1$.
 As in~\cite{FernauGHHVW19},
 we set $\Sigma_{i,j} = \{ x \in \Sigma : \mu(i, x) = j \}$.
 If the cases in Proposition~\ref{prop:structure_simp_idem_binary} apply, and
 if so, which case,
 can be checked in polynomial time, as we only have to check
 that one letter is a simple idempotent and the other letter permutes
 the states in a single cycle or two cycles with the restrictions
 as written in Proposition~\ref{prop:structure_simp_idem_binary}.
 
 So, we can assume $\mathcal A$
 has one of the two forms as written in Proposition~\ref{prop:structure_simp_idem_binary}.
 Without loss of generality, we assume $a$
 is the idempotent letter and $b$ the cyclic permutation of the states.
 Set $n = |Q|$. We also assume $n > 4$, if $n \le 4$, then the cases for $\mathcal A$
 can be checked in constant time for a given (fixed) constraint language
 $L(\mathcal B)$.

 Next, we only handle the first case of Proposition~\ref{prop:structure_simp_idem_binary} by case analysis. The 
 other case can be handled similarly. Further, let $t \in Q$
 be the state as written in the first case of Proposition~\ref{prop:structure_simp_idem_binary}.
 
 Further, in this sketch, we only handle the case
 that the strongly connected components of $\mathcal B$
 are $\{1\}$ and $\{2,3\}$
 and the subautomaton
 between the states $\{2,3\}$
 is one of the automata
 listed in Table~\ref{thm:simple_idempotents_binary_in_P}. These are the most difficult
 cases, for the remaining cases
 of $\mathcal B$ we refer to the full
 proof in the appendix.

 Next, we handle each of these cases separately.
 We show that for each case, either the input
 automaton with the assumed form is synchronizing (and hence 
 the problem is only to check that $\mathcal A$
 has the stated form, which can be done in polynomial
 time as said in the beginning of this proof),
 or we have another simple to check condition, 
 like if $n$ is even or odd.

\begin{table}[ht] 
    \centering
 
\begin{tabular}{c|c|c|c}
 Case & $|\Sigma_{2,2} + \Sigma_{2,3}| = 2$ & Case & $|\Sigma_{3,3}+\Sigma_{3,2}|=2$ \\ \hline
 
 1 & \scalebox{.7}{
\begin{tikzpicture}[>=latex',shorten >=1pt,node distance=1.5cm,on grid,auto,baseline=0]
 \tikzset{every state/.style={minimum size=1pt}}
 
 \node[state] (1) {};
 \node[state, accepting] (2) [right of=1] {};
 
 \path[->] (1) edge [loop above] node {$a$} (1);
 \path[->] (1) edge [bend left]  node {$b$} (2);
 \path[->] (2) edge [bend left]  node {$b$} (1);
\end{tikzpicture}} 

 & 2 & 
\scalebox{.7}{
\begin{tikzpicture}[>=latex',shorten >=1pt,node distance=1.5cm,on grid,auto,baseline=0]
 \tikzset{every state/.style={minimum size=1pt}}
 
 \node[state] (1) {};
 \node[state, accepting] (2) [right of=1] {};
 
 \path[->] (2) edge [loop above] node {$a$} (2);
 \path[->] (1) edge [bend left]  node {$a$} (2);
 \path[->] (2) edge [bend left]  node {$b$} (1);
\end{tikzpicture}} 
\\ \hline

 3 & \scalebox{.7}{
\begin{tikzpicture}[>=latex',shorten >=1pt,node distance=1.5cm,on grid,auto,baseline=0]
 \tikzset{every state/.style={minimum size=1pt}}
 
 \node[state] (1) {};
 \node[state, accepting] (2) [right of=1] {};
 
 \path[->] (1) edge [loop above] node {$a$} (1);
 \path[->] (1) edge [bend left]  node {$b$} (2);
 \path[->] (2) edge [bend left]  node {$a$} (1);
\end{tikzpicture}}

 & 4 & 
\scalebox{.7}{
\begin{tikzpicture}[>=latex',shorten >=1pt,node distance=1.5cm,on grid,auto,baseline=0]
 \tikzset{every state/.style={minimum size=1pt}}
 
 \node[state] (1) {};
 \node[state, accepting] (2) [right of=1] {};
 
 \path[->] (2) edge [loop above] node {$b$} (2);
 \path[->] (1) edge [bend left]  node {$a$} (2);
 \path[->] (2) edge [bend left]  node {$a$} (1);
\end{tikzpicture}} 
\\ \hline

 5 & \scalebox{.7}{
\begin{tikzpicture}[>=latex',shorten >=1pt,node distance=1.5cm,on grid,auto,baseline=0]
 \tikzset{every state/.style={minimum size=1pt}}
 
 \node[state] (1) {};
 \node[state, accepting] (2) [right of=1] {};
 
 \path[->] (1) edge [loop above] node {$b$} (1);
 \path[->] (1) edge [bend left]  node {$a$} (2);
 \path[->] (2) edge [bend left]  node {$b$} (1);
\end{tikzpicture}} 

 & 6 & 
\scalebox{.7}{
\begin{tikzpicture}[>=latex',shorten >=1pt,node distance=1.5cm,on grid,auto,baseline=0]
 \tikzset{every state/.style={minimum size=1pt}}
 
 \node[state] (1) {};
 \node[state, accepting] (2) [right of=1] {};
 
 \path[->] (2) edge [loop above] node {$a$} (2);
 \path[->] (1) edge [bend left]  node {$b$} (2);
 \path[->] (2) edge [bend left]  node {$b$} (1);
\end{tikzpicture}} 
\\ \hline

 7 & \scalebox{.7}{
\begin{tikzpicture}[>=latex',shorten >=1pt,node distance=1.5cm,on grid,auto,baseline=0]
 \tikzset{every state/.style={minimum size=1pt}}
 
 \node[state] (1) {};
 \node[state, accepting] (2) [right of=1] {};
 
 \path[->] (1) edge [loop above] node {$b$} (1);
 \path[->] (1) edge [bend left]  node {$a$} (2);
 \path[->] (2) edge [bend left]  node {$a$} (1);
\end{tikzpicture}} 

 & 8 & 
\scalebox{.7}{
\begin{tikzpicture}[>=latex',shorten >=1pt,node distance=1.5cm,on grid,auto,baseline=0]
 \tikzset{every state/.style={minimum size=1pt}}
 
 \node[state] (1) {};
 \node[state, accepting] (2) [right of=1] {};
 
 \path[->] (2) edge [loop above] node {$b$} (2);
 \path[->] (1) edge [bend left]  node {$b$} (2);
 \path[->] (2) edge [bend left]  node {$a$} (1);
\end{tikzpicture}} 
\\ \hline

 9 & \scalebox{.7}{
\begin{tikzpicture}[>=latex',shorten >=1pt,node distance=1.5cm,on grid,auto,baseline=0]
 \tikzset{every state/.style={minimum size=1pt}}
 
 \node[state] (1) {};
 \node[state, accepting] (2) [right of=1] {};
 
 \path[->] (1) edge [bend left]  node {$a,b$} (2);
 \path[->] (2) edge [bend left]  node {$a$} (1);
\end{tikzpicture}}

 & 10 & 
\scalebox{.7}{
\begin{tikzpicture}[>=latex',shorten >=1pt,node distance=1.5cm,on grid,auto,baseline=0]
 \tikzset{every state/.style={minimum size=1pt}}
 
 \node[state] (1) {};
 \node[state, accepting] (2) [right of=1] {};
 
 \path[->] (1) edge [bend left]  node {$a$} (2);
 \path[->] (2) edge [bend left]  node {$a,b$} (1);
\end{tikzpicture}} 
\\ \hline

 11 & \scalebox{.7}{
\begin{tikzpicture}[>=latex',shorten >=1pt,node distance=1.5cm,on grid,auto,baseline=0]
 \tikzset{every state/.style={minimum size=1pt}}
 
 \node[state] (1) {};
 \node[state, accepting] (2) [right of=1] {};
 
 \path[->] (1) edge [bend left]  node {$a,b$} (2);
 \path[->] (2) edge [bend left]  node {$b$} (1);
\end{tikzpicture}} 

 & 12 & 
\scalebox{.7}{
\begin{tikzpicture}[>=latex',shorten >=1pt,node distance=1.5cm,on grid,auto,baseline=0]
 \tikzset{every state/.style={minimum size=1pt}}
 
 \node[state] (1) {};
 \node[state, accepting] (2) [right of=1] {};

 \path[->] (1) edge [bend left]  node {$b$} (2);
 \path[->] (2) edge [bend left]  node {$a,b$} (1);
\end{tikzpicture}} \\ \hline

\end{tabular}
 \caption{Cases for a partial subautomaton between the states $\{2,3\}$
  that is not complete. See the proof of Theorem~\ref{thm:simple_idempotents_binary_in_P}
  for details.}
 \label{tab:inner_automata_three_states}
  \vspace{-8mm}
\end{table}

 We handle the cases as numbered in Table~\ref{thm:simple_idempotents_binary_in_P}.

\begin{enumerate}

\item \label{case:a_bb_star} In this case $L(\mathcal B_{2,\{2\}}) = (a+bb)^*$.   
 Let $s \in Q \setminus\{t\}$ be the state with $\delta(s, a) = \delta(t, a) = t$.
 If $n$ is odd, then $|Q \setminus \{t\}|$ is even
 and the single cycle induced by $b$ on the states in $Q \setminus \{t\}$
 splits into two cycles for the word $bb$, i.e., we have precisely
 two disjoint subsets $A, B \subseteq Q \setminus\{t\}$
 of equal size such that the
 states in one subset can be mapped onto each other by a word in $(bb)^*$
 but we cannot map states between those subset by a word from $(bb)^*$.
 Suppose, without loss of generality, that $s \in A$.
 Then for each $q \in A$ we have $\delta(q, a) = q$
 and $\delta(q, bb) \in A$. 
 Hence, we cannot map a state from $A$ to $s$, and so not to $t$, the unique 
 synchronizing state. As $n \ge 5$ and 
 for every $w \in \Sigma_{1,1}^*\Sigma_{1,2}$ 
 we have $|\delta(Q, w)| \ge n - 1$,
 we must have $\delta(A, w) \ne \emptyset$
 for every $w \in \Sigma_{1,1}^*\Sigma_{1,2}$. So, $\mathcal A$ cannot be synchronized by
 a word from $L(\mathcal B_{1,\{2\}}) = \Sigma_{1,1}^* \Sigma_{1,2} (a + bb)^*$.
 If $n$ is even, than $bb$ also permutes the states in $Q \setminus \{t\}$
 in a single cycle
 and the word $a(bba)^{n-2}$ synchronizes $\mathcal A$.
 So, picking any $u \in \{a,b\}^*$
 with $\mu(1, u) = 2$, the word $u(bba)^{n-1} \in L(\mathcal B)$
 synchronizes $\mathcal A$.
 
\item In this case $L(\mathcal B_{2,\{2\}}) = (aa^*b)^*$.
 The word $(ab)^{n-1}$ synchronizes $\mathcal A$.
 Pick any $u \in L(\mathcal B_{1,\{2\}})$,
 then $u(ab)^{n-1} \in L(\mathcal B)$
 synchronizes $\mathcal A$.

\item In this case $L(\mathcal B_{2,\{2\}}) = (a+ba)^*$.
  For every $q \ne t$ we have $\delta(q, ba) = \delta(q, b)$
  and the word $a(ba)^{n-2}$ synchronizes $\mathcal A$.
  Let $u \in \{a,b\}^*$ be a word with $\mu(1, u) = 2$,
  Then $\mu(1, ua(ba)^{n-2}) = 2$ synchronizes $\mathcal A$. 

\item In this case $L(\mathcal B_{2,\{2\}}) = (ab^*a)^*$.
 The word $(aba)^{n-2}$ synchronizes $\mathcal A$ (as
 $a$ is idempotent, it has the same effect as the synchronizing word $a(ba)^{n-2}$
 on the states of $\mathcal A$).
 Pick any $u \in L(\mathcal B_{1,\{2\}})$,
 then $u(aba)^{n-2}\in L(\mathcal B)$
 synchronizes $\mathcal A$.

\item In this case $L(\mathcal B_{2,\{2\}}) = (b+ab)^*$.
 The word $(ab)^{n-1} = a(ba)^{n-2}b$ synchronizes $\mathcal A$,
 as $a(ba)^{n-2}$ synchronizes $\mathcal A$.
 Pick any $u \in L(\mathcal B_{1,\{2\}})$,
 then $u(aba)^{n-2}\in L(\mathcal B)$
 synchronizes $\mathcal A$.

\item In this case $L(\mathcal B_{2,\{2\}}) = (ba^*b)^*$.
 As $a$ is idempotent, the words $bab$ and $ba^ib$ for $i > 1$
 all have the same effect on the states of $\mathcal A$.
 So, only the language $\{ bb, bab \}^* \subseteq L(\mathcal B_{2,\{2\}})$
 is of relevance for the question of synchronizability of $\mathcal A$
 with respect to the constraint language $L(\mathcal B)$.
 We have $|\delta(Q, bab)| = n - 1$
 and if $n$ is odd, we can argue as in Case~\ref{case:a_bb_star}
 that $\mathcal A$ can not be synchronized
 by a word from $L(\mathcal B_{1,\{2\}})$.
 If $n$ is even, the word $bb$ induces a single cycle on the
 states in $Q \setminus \{t\}$
 and $(bbbab)^{n-1}$ synchronizes $\mathcal A$.
 As in previous cases, by appending a suitable prefix,
 we can construct a synchronizing word in $L(\mathcal B)$.

\item In this case $L(\mathcal B_{2,\{2\}}) = (b+aa)^*$.
 Then, choosing any $u \in L(\mathcal B_{1,\{2\}})$,
 the word $u(baa)^{n-1}$ synchronizes $\mathcal A$.

\item In this case $L(\mathcal B_{2,\{2\}}) = (bb^*a)^*$.
 As $(ba)^{n-1} \in L(\mathcal B_{2,\{2\}})$ synchronizes $\mathcal A$
 we can pick any $u \in \{a,b\}^*$ such that $\mu(1, u) = 2$
 and have the synchronizing word $u(ba)^{n-1} \in L(\mathcal B)$ 
 for $\mathcal A$. 

\item Here $L(\mathcal B_{2,\{2\}}) = (aa+ba)^*$.
 Using that $aa(ba)^{n-2} \in L(\mathcal B_{2,\{2\}})$ synchronizes $\mathcal A$
 as in the previous cases by appending a suitable prefix from $L(\mathcal B_{1,\{2\}})$.
\item Here $L(\mathcal B_{2,\{2\}}) = (aa+ab)^*$.
 Using that $aa(ab)^{n-2}aa \in L(\mathcal B_{2,\{2\}})$ synchronizes $\mathcal A$
 as in the previous cases by appending a suitable prefix from $L(\mathcal B_{1,\{2\}})$.
 
\item Here $L(\mathcal B_{2,\{2\}}) = (ab+bb)^*$.
 Using that $a(ba)^{n-2}b = (ab)^{n-1} \in L(\mathcal B_{2,\{2\}})$ synchronizes $\mathcal A$
 as in the previous cases by appending a suitable prefix from $L(\mathcal B_{1,\{2\}})$.
 
\item Here $L(\mathcal B_{2,\{2\}}) = (ba+bb)^*$.
 Using that $ba(ba)^{n-2} \in L(\mathcal B_{2,\{2\}})$ synchronizes $\mathcal A$
 as in the previous cases by appending a suitable prefix from $L(\mathcal B_{1,\{2\}})$.\qedhere
\end{enumerate}
\end{proofsketch}

\begin{proof}
 The cases $|\Sigma| \le 1$
 or $|P| \le 2$ and $|\Sigma| = 2$
 are polynomial time decidable in general
 as shown in~\cite[Corollary 9 \& Theorem 24]{FernauGHHVW19}.
 So, we can suppose $\Sigma = \{a,b\}$
 and $P = \{1,2,3\}$ with $p_0 = 1$.
 As in~\cite{FernauGHHVW19},
 we set $\Sigma_{i,j} = \{ x \in \Sigma : \mu(i, x) = j \}$.
 If the cases in Proposition~\ref{prop:structure_simp_idem_binary} apply, and
 if so, which case,
 can be checked in polynomial time, as we only have to check
 that one letter is a simple idempotent and the other letter permutes
 the states in a single cycle or two cycles with the restrictions
 as written in Proposition~\ref{prop:structure_simp_idem_binary}.
 
 So, we can assume $\mathcal A$
 has one of the two forms as written in Proposition~\ref{prop:structure_simp_idem_binary}.
 Without loss of generality, we assume $a$
 is the idempotent letter and $b$ the cyclic permutation of the states.
 Set $n = |Q|$. We also assume $n > 4$, if $n \le 4$, then the cases for $\mathcal A$
 can be checked in constant time for a given (fixed) constraint language
 $L(\mathcal B)$.

 Next, we only handle the first case of Proposition~\ref{prop:structure_simp_idem_binary} by case analysis. The 
 other case can be handled similarly. Further, let $t \in Q$
 be the state as written in the first case of Proposition~\ref{prop:structure_simp_idem_binary}.
 
 We can assume every state in $P$ is reachable from $p_0$
 and also that at least one state in $\{2,3\}$
 is reachable from the other state in this set (otherwise, it is a union
 of two constraint languages over two-state constraint automata,
 each giving a constraint problem in $\PTIME$, and so, also their union
 gives a problem in \PTIME, as we can check each constraint individually,
 see also~\cite[Lemma 13]{FernauGHHVW19}).
 Without loss of generality, we assume $\Sigma_{2,3} \ne \emptyset$.
 If the states in $P$ form a strongly connected component, then in $\mathcal B$
 we can map every state back to the starting state and
 by~\cite[Theorem 17]{FernauGHHVW19} the constrained problem
 is solvable in polynomial time.
 If every state forms its own strongly connected component,
 we can assume $3 \in F$ (otherwise, it reduces to the two-state case).
 As every state is reachable from $p_0$,
 we then have $\Sigma_{2,1} = \Sigma_{3,1} = \emptyset$.
 Then
 \[
  L(\mathcal B) = \Sigma_{1,1}^* \Sigma_{1,2} \Sigma_{2,2}^* \Sigma_{2,3} \Sigma_{3,3}^*
   \cup \Sigma_{1,1}^* \Sigma_{1,3} \Sigma_{3,3}^*.
 \] 
 and either $\Sigma_{1,2}\cdot \Sigma_{2,3} \ne \emptyset$
 or $\Sigma_{1,3} \ne \emptyset$.
 If $\Sigma_{3,3} = \{a,b\}$, then $\mathcal A$
 has a synchronizing word in $L(\mathcal B)$ iff it has a synchronizing
 word at all. For if we take an arbitrary synchronizing
 word $u$, we can take any word $v$
 with $\mu(1, v) = 3$ and then $vu$ is synchronizing for $\mathcal A$ (this is actually
 a special case of~\cite[Theorem 15]{FernauGHHVW19}).
 As the unconstrained synchronization problem is polynomial time solvable,
 in the case $\Sigma_{3,3} = \{a,b\}$
 the constrained problem is also polynomial time solvable.
 As $\mathcal B$ is deterministic,
 we can furthermore assume $|\Sigma_{i,i}| \le 1$
 for all $i \in \{1,2,3\}$ (if, for example $|\Sigma_{2,2}| = 2$,
 this implies $\Sigma_{2,3} = \emptyset$, which reduces
 again to the two-state case).
 However, $a$ is idempotent, and so, if, for example $\Sigma_{2,2} = \{a\}$,
 which implies $a \notin \Sigma_{2,3}$,
 we can suppose the corresponding self-loop in $\mathcal B$
 is traversed at most two times.
 Considering all such cases, we can deduce that every synchronizing word
 has the same effect as a synchronizing word where
 the number of $a$'s in it is bounded by three. 
 But as $n > 4$, by Lemma~\ref{lem:each_sync_word_contains_a}
 every synchronizing word $w$
 has to fulfill $|w| \ge 4$
 and so, for these constraint languages every input semi-automaton $\mathcal A$
 with at least five states has no synchronizing word in $L(\mathcal B)$.

 So, from now on we can assume that we have one strongly connected
 component with precisely two states and one with a single state.
 We handle the case that $\{2,3\}$ is a strongly connected component here,
 the other case that $\{1,2\}$ is a strongly connected component can
 be handled similarly. 
 As every state is reachable from $p_0$, 
 under this assumption we must have $\Sigma_{2,1} = \Sigma_{3,1} = \emptyset$.
 As either $\Sigma_{1,2} \ne \emptyset$
 or $\Sigma_{1,3}\ne \emptyset$, we must have $|\Sigma_{1,1}|\le 1$.
 Note that $F \cap \{2,3\} \ne \emptyset$,
 for otherwise the constrained problem is equivalent
 to the problem over a one-state automaton, which
 is in \PTIME\  (see~\cite[Corollary 9]{FernauGHHVW19}).
 Further, we can assume $F = \{2\}$,
 as if $\mu(p_0, u) \in F$,
 by the assumptions there exists a word $w$
 such that $\mu(p_0, uw) = 2$
 and conversely, if $\mu(p_0, u) = 2$,
 there exists a word $w$ such that $\mu(p_0, uw) \in F$
 as $\{2,3\} \cap F \ne \emptyset$ and both states are reachable
 from $p_0$. As we can append arbitrary words to a synchronizing
 word and still have a synchronizing word,
 asking for a synchronizing word in $L(\mathcal B)$
 under the stated assumptions is equivalent 
 as asking for a synchronizing word with final state set $\{2\}$.

 We can focus only on the case $\Sigma_{1,3} = \emptyset$,
 as showing that the problem is in \PTIME \todo{hier etwas genauer?}
 in this case, the case $\Sigma_{1,3} \ne \emptyset$
 can be handled symmetrically under
 the assumption $\Sigma_{1,2} \cdot \Sigma_{2,3} = \emptyset$
 and the general case can be written as a union of both cases
 with an appropriate distribution of the final states,
 which gives that polynomial time solvability by using~\cite[Lemma 13]{FernauGHHVW19}.

 For $p \in P$ and $E \subseteq P$, we write $\mathcal B_{p,E} = (\Sigma, P, \mu, p, E)$
 for the PDFA that results from $\mathcal B$ by changing the start state to $p$
 and the set of final states to $E$.

 So, with the assumption that $\{2,3\}$
 is a strongly connected component, 
 and hence $\Sigma_{2,1} = \Sigma_{3,1}  = \emptyset$,
 and $\Sigma_{1,3} = \emptyset$,
 we can deduce that every word
 in $\Sigma_{1,1}^*\Sigma_{1,2}$ (note that
 this language does not equal $L(\mathcal B_{1,\{2\}})$)
 can map at most two states to a single
 state, and one of these states must be $t$,
 and every other state is mapped to distinct states, 
 i.e., for every $w \in \Sigma_{1,1}^*\Sigma_{1,2}$,
 $|\delta(Q, w)| \ge n - 1$
 and if $|\delta(Q, w)| = n - 1$,
 then two states, where one is $t$ itself,
 are mapped to $t$.
 If $\Sigma_{1,1} = \emptyset$
 this is clear as $w \in \{a,b\}$.
 If $\Sigma_{1,1} = \{a\}$,
 we have $w \in a^*b$
 and, as $a$ is idempotent, every such word has the same effect
 on the state set $Q$ of $\mathcal A$ as the two words $b$
 or $ab$.
 If $\Sigma_{1,1} = \{b\}$, then $w \in b^*a$
 and $\delta(Q, b^ia) = \delta(Q, a)$ for all $i \ge 1$,
 as $b$ permutes all states. The claim
 about $t$ follows by choice of this state, see
 Proposition~\ref{prop:structure_simp_idem_binary}.

 Now, note that if the subautomaton between
 the states $\{2,3\}$ is complete,
 then if we have any synchronizing word $u$,
 by choosing $w, v$
 such that $\mu(p_0, w) = 2$
 and $\mu(2, uv) = 2$, which can
 be done by the assumptions,
 we have a synchronizing word $wuv \in L(\mathcal B)$,
 and so the problem to find 
 a synchronizing word in $L(\mathcal B)$
 is equivalent to the unconstrained synchronization problem,
 which is solvable in polynomial time.
 Hence, we can assume at least one transition 
 between the states in $\{2,3\}$
 is undefined.
 Further, if strictly more than one transition
 between these states is not defined,
 then, as these states form a strongly connected
 component, we have $\Sigma_{2,2} = \Sigma_{3,3} = \emptyset$
 and in this
 case, if $\Sigma_{2,3} = \Sigma_{3,2}$,
 we either have $L(\mathcal B) = \Sigma_{1,1}^*\Sigma_{1,2}(bb)^*$
 or $L(\mathcal B) = \Sigma_{1,1}^*\Sigma_{1,2}(aa)^*$,
 and it is easy to see that for $n > 4$
 we cannot have a synchronizing word for $\mathcal A$
 in these languages.
 If $\Sigma_{2,3} \ne \Sigma_{3,2}$,
 then $L(\mathcal B) = \Sigma_{1,1}^*\Sigma_{1,2}(ab)^*$
 or $L(\mathcal B) = \Sigma_{1,1}^*\Sigma_{1,2}(ba)^*$
 and as $a(ba)^{n-2}$ synchronizes $\mathcal A$
 if it has the form as stated in case one
 of Proposition~\ref{prop:structure_simp_idem_binary},
 by choosing
 $u \in \Sigma_{1,1}^* \Sigma_{1,2}$,
 the word $ua(ba)^{n-2}b$ synchronizes $\mathcal A$
 and is in $\Sigma_{1,1}^*\Sigma_{1,2}(ab)^*$
 and $uba(ba)^{n-2}$ synchronizing $\mathcal A$
 and is in $\Sigma_{1,1}^*\Sigma_{1,2}(ab)^*$.
 
 So, the only cases left for the subautomaton
 between the states $\{2,3\}$
 are listed in Table~\ref{thm:simple_idempotents_binary_in_P},
 which have been handled in the main text.
 
\end{proof}

\section{Conclusion}

\todo{noch ergebnisse und offene fragen (trichotomy) zu allg constrained sync diskutieren?}

Here, we have shown that for commutative input automata, the constrained
synchronization problem is tractable for every regular constraint language.
Hence, this case is the first case of input automata where the question of the computational
complexity is settled completely. In~\cite{DBLP:conf/dlt/Hoffmann21}, it was shown that for
weakly acyclic automata, the problem is always in \NP, and there exists \NP-hard
instances and tractable instance of the problem. However, it was not shown that these are the only complexities
that can arise.

Note that, on the other side, for commutative regular constraint languages,
we can realize \PSPACE-complete problems (for example for $\{a,b\}^*c\{a,b\}^*)$),
\NP-complete (for $a^*ba^*ba^*$) or problems in \PTIME.
In fact, a full classification giving a trichotomy that only
these three complexities arise has been shown for commutative regular constraints~\cite{HoffmannCocoonExtended}.

Furthermore, in the present work, we have started the investigation of the constrained
synchronization problem for input automata with simple idempotents.
Contrary to the results for commutative automata, which are conclusive, in the case of automata
with simple idempotents, the present results are only the starting point. They entail the first non-trivial
instances of these automata -- for example the automata from the \v{C}ern\'y family~\cite{Cerny64,Vol2008}
giving a lower bound for the length of a shortest synchronizing word are binary automata
with simple idempotents -- but also the first instances of constraint automata that can realize \PSPACE-complete or \NP-complete problems in the general case~\cite{FernauGHHVW19}.
In this respect, it is remarkable that the complexity drops to being polynomial time solvable in this case.
However, it is unclear what happens for larger alphabets and different constraint languages when only automata with simple idempotents are considered as input.

\smallskip \noindent {\footnotesize
\textbf{Acknowledgement.} I thank  anonymous reviewers
of the extended version (submitted) of~\cite{HoffmannCocoon20}, whose
feedback influenced Section~\ref{sec:comm_input}.
I am also grateful to the reviewers of a previous version.
One reviewer suggested Proposition~\ref{prop:intersection_waa_com} and its proof
and noted that it greatly simplifies the arguments
from Section~\ref{sec:comm_input}. I am really pleased by
this result, which greatly simplified and improved the presentation.
}

\bibliography{ms,cerny} 

\newcommand{\etalchar}[1]{$^{#1}$}
\begin{thebibliography}{FGH{\etalchar{+}}19}

\bibitem[CLRS09]{DBLP:books/daglib/0023376}
Thomas~H. Cormen, Charles~E. Leiserson, Ronald~L. Rivest, and Clifford Stein.
\newblock {\em Introduction to Algorithms, 3rd Edition}.
\newblock {MIT} Press, 2009.

\bibitem[FGH{\etalchar{+}}19]{FernauGHHVW19}
Henning Fernau, Vladimir~V. Gusev, Stefan Hoffmann, Markus Holzer, Mikhail~V.
  Volkov, and Petra Wolf.
\newblock Computational complexity of synchronization under regular
  constraints.
\newblock In Peter Rossmanith, Pinar Heggernes, and Joost{-}Pieter Katoen,
  editors, {\em 44th International Symposium on Mathematical Foundations of
  Computer Science, {MFCS} 2019, August 26-30, 2019, Aachen, Germany}, volume
  138 of {\em LIPIcs}, pages 63:1--63:14. Schloss Dagstuhl - Leibniz-Zentrum
  f{\"{u}}r Informatik, 2019.

\bibitem[FH19]{FernauHoffmann19}
Henning Fernau and Stefan Hoffmann.
\newblock Extensions to minimal synchronizing words.
\newblock {\em J. Autom. Lang. Comb.}, 24(2-4):287--307, 2019.

\bibitem[Hof21]{DBLP:conf/dlt/Hoffmann21a}
Stefan Hoffmann.
\newblock Constrained synchronization and subset synchronization problems for
  weakly acyclic automata.
\newblock In Nelma Moreira and Rog{\'{e}}rio Reis, editors, {\em Developments
  in Language Theory - 25th International Conference, {DLT} 2021, Porto,
  Portugal, August 16-20, 2021, Proceedings}, volume 12811 of {\em Lect. Notes
  Comp. Sci.}, pages 204--216. Springer, 2021.

\bibitem[Mar09]{Mar2009b}
Pavel Martyugin.
\newblock Complexity of problems concerning reset words for some partial cases
  of automata.
\newblock {\em Acta Cybernetica}, 19(2):517--536, 2009.

\end{thebibliography}


\begin{thebibliography}{10}

\bibitem{AlmeidaS09}
Jorge Almeida and Benjamin Steinberg.
\newblock Matrix mortality and the {Cern{\'{y}}}-{Pin} conjecture.
\newblock In Volker Diekert and Dirk Nowotka, editors, {\em Developments in
  Language Theory, 13th International Conference, {DLT} 2009, Stuttgart,
  Germany, June 30 - July 3, 2009. Proceedings}, volume 5583 of {\em Lecture
  Notes in Computer Science}, pages 67--80. Springer, 2009.

\bibitem{Alves2020}
Lucas~V.R. Alves and Patrícia~N. Pena.
\newblock Synchronism recovery of discrete event systems.
\newblock {\em IFAC-PapersOnLine}, 53(2):10474--10479, 2020.
\newblock 21th IFAC World Congress.

\bibitem{AnanichevVolkov03}
Dimitry~S. Ananichev and Mikhail~V. Volkov.
\newblock Synchronizing monotonic automata.
\newblock In Z.~{\'E}sik and Z.~F{\"u}l{\"o}p, editors, {\em DLT 2003}, number
  2710 in LNCS, pages 111--121, Berlin, Heidelberg, 2003. Springer.

\bibitem{AnanichevV05}
Dimitry~S. Ananichev and Mikhail~V. Volkov.
\newblock Synchronizing generalized monotonic automata.
\newblock {\em Theor. Comput. Sci.}, 330(1):3--13, 2005.

\bibitem{ArCamSt2017}
João Araújo, Peter~J. Cameron, and Benjamin Steinberg.
\newblock Between primitive and 2-transitive: Synchronization and its friends.
\newblock {\em EMS Surveys in Math. Sciences}, 4(2):101--184, 2017.

\bibitem{ArnoldS06}
Fredrick Arnold and Benjamin Steinberg.
\newblock Synchronizing groups and automata.
\newblock {\em Theor. Comput. Sci.}, 359(1-3):101--110, 2006.

\bibitem{Arrighi2021}
Emmanuel Arrighi, Henning Fernau, Stefan Hoffmann, Markus Holzer, Isma\"el
  Jecker, Mateus de~Oliveira~Oliveira, and Petra Wolf.
\newblock On the complexity of {Intersection Non-emptiness} for star-free
  language classes.
\newblock (in preparation).

\bibitem{BerlinkovS16}
Mikhail~V. Berlinkov and Marek Szykula.
\newblock Algebraic synchronization criterion and computing reset words.
\newblock {\em Inf. Sci.}, 369:718--730, 2016.

\bibitem{DBLP:books/daglib/0034521}
Christos~G. Cassandras and St{\'{e}}phane Lafortune.
\newblock {\em Introduction to Discrete Event Systems, Second Edition}.
\newblock Springer, 2008.

\bibitem{Cerny64}
J\'an {\v C}ern{\'y}.
\newblock Pozn{\'a}mka k. homog\'ennym experimentom s konecn{\'y}mi automatmi.
\newblock {\em Mat. fyz. {\v c}as SAV}, 14:208--215, 1964.

\bibitem{CernyPirickaRosenauerova71}
J\'an {\v C}ern{\'y}, Alica Pirick\'a, and Blanka Rosenauerova.
\newblock On directable automata.
\newblock {\em Kybernetica}, 7:289--298, 1971.

\bibitem{DBLP:journals/algorithmica/ChenI95}
Yui{-}Bin Chen and Doug Ierardi.
\newblock The complexity of oblivious plans for orienting and distinguishing
  polygonal parts.
\newblock {\em Algorithmica}, 14(5):367--397, 1995.

\bibitem{DBLP:journals/et/ChoJSP93}
Hyunwoo Cho, Seh{-}Woong Jeong, Fabio Somenzi, and Carl Pixley.
\newblock Synchronizing sequences and symbolic traversal techniques in test
  generation.
\newblock {\em J. Electron. Test.}, 4(1):19--31, 1993.

\bibitem{Don16}
Henk Don.
\newblock The {\v{c}}ern{\'{y}} conjecture and 1-contracting automata.
\newblock {\em Electron. J. Comb.}, 23(3):P3.12, 2016.

\bibitem{Dubuc96}
Louis Dubuc.
\newblock Les automates circulaires biais\'es verifient la conjecture de
  \v{C}ern\'{y}.
\newblock {\em Inform. Theor. Appl.}, 30:495--505, 1996.

\bibitem{Dubuc98}
Louis Dubuc.
\newblock Les automates circulaires et la conjecture de \v{C}ern\'{y}.
\newblock {\em Inform. Theor. Appl.}, 32:21--34, 1998.

\bibitem{Eppstein90}
David Eppstein.
\newblock Reset sequences for monotonic automata.
\newblock {\em SIAM J. Comput.}, 19:500--510, 1990.

\bibitem{DBLP:journals/trob/ErdmannM88}
Michael~A. Erdmann and Matthew~T. Mason.
\newblock An exploration of sensorless manipulation.
\newblock {\em {IEEE} J. Robotics Autom.}, 4(4):369--379, 1988.

\bibitem{Pin81a}
Jean \'{E}ric Pin.
\newblock Le probl\`eme de la synchronisation et la conjecture de {\v
  c}ern{\'y}.
\newblock In A.~De~luca, editor, {\em Non-commutative structures in algebra and
  geometric combinatorics}, volume 109 of {\em Quaderni de la Ricerca
  Scientifica}, pages 37--48. CNR, Roma, 1981.

\bibitem{Pin83a}
Jean \'{E}ric Pin.
\newblock On two combinatorial problems arising from automata theory.
\newblock {\em Annals of Discrete Mathematics}, 17:535--548, 1983.

\bibitem{FernauGHHVW19}
Henning Fernau, Vladimir~V. Gusev, Stefan Hoffmann, Markus Holzer, Mikhail~V.
  Volkov, and Petra Wolf.
\newblock Computational complexity of synchronization under regular
  constraints.
\newblock In Peter Rossmanith, Pinar Heggernes, and Joost{-}Pieter Katoen,
  editors, {\em 44th International Symposium on Mathematical Foundations of
  Computer Science, {MFCS} 2019, August 26-30, 2019, Aachen, Germany}, volume
  138 of {\em LIPIcs}, pages 63:1--63:14. Schloss Dagstuhl - Leibniz-Zentrum
  f{\"{u}}r Informatik, 2019.

\bibitem{FernauHoffmann19}
Henning Fernau and Stefan Hoffmann.
\newblock Extensions to minimal synchronizing words.
\newblock {\em J. Autom. Lang. Comb.}, 24(2-4):287--307, 2019.

\bibitem{FernauHoffmannWehar2021}
Henning Fernau, Stefan Hoffmann, and Michael Wehar.
\newblock Finite automata intersection non-emptiness: Parameterized complexity
  revisited.
\newblock {\em CoRR}, abs/2108.05244, 2021.
\newblock \href {http://arxiv.org/abs/2108.05244} {\path{arXiv:2108.05244}}.

\bibitem{FluGro2006}
Jörg Flum and Martin Grohe.
\newblock {\em Parameterized Complexity Theory}.
\newblock Springer, 2006.

\bibitem{Frankl82}
Peter Frankl.
\newblock An extremal problem fro two families of sets.
\newblock {\em Eur. J. Comb.}, 3:125--127, 1982.

\bibitem{DBLP:journals/algorithmica/Goldberg93}
Kenneth~Y. Goldberg.
\newblock Orienting polygonal parts without sensors.
\newblock {\em Algorithmica}, 10(2-4):210--225, 1993.

\bibitem{HoffmannCocoon20}
Stefan Hoffmann.
\newblock Computational complexity of synchronization under regular commutative
  constraints.
\newblock In Donghyun Kim, R.~N. Uma, Zhipeng Cai, and Dong~Hoon Lee, editors,
  {\em Computing and Combinatorics - 26th International Conference, {COCOON}
  2020, Atlanta, GA, USA, August 29-31, 2020, Proceedings}, volume 12273 of
  {\em Lecture Notes in Computer Science}, pages 460--471. Springer, 2020.

\bibitem{HoffmannCocoonExtended}
Stefan Hoffmann.
\newblock Constrained synchronization and commutative languages.
\newblock {\em Theoretical Computer Science}, 2021.
\newblock (in press).
\newblock \href {https://doi.org/10.1016/j.tcs.2021.08.030}
  {\path{doi:10.1016/j.tcs.2021.08.030}}.

\bibitem{DBLP:conf/dlt/Hoffmann21a}
Stefan Hoffmann.
\newblock Constrained synchronization and subset synchronization problems for
  weakly acyclic automata.
\newblock In Nelma Moreira and Rog{\'{e}}rio Reis, editors, {\em Developments
  in Language Theory - 25th International Conference, {DLT} 2021, Porto,
  Portugal, August 16-20, 2021, Proceedings}, volume 12811 of {\em Lect. Notes
  Comp. Sci.}, pages 204--216. Springer, 2021.

\bibitem{DBLP:conf/dlt/Hoffmann21}
Stefan Hoffmann.
\newblock State complexity of projection on languages recognized by permutation
  automata and commuting letters.
\newblock In Nelma Moreira and Rog{\'{e}}rio Reis, editors, {\em Developments
  in Language Theory - 25th International Conference, {DLT} 2021, Porto,
  Portugal, August 16-20, 2021, Proceedings}, volume 12811 of {\em Lect. Notes
  Comp. Sci.}, pages 192--203. Springer, 2021.
\newblock \href {https://doi.org/10.1007/978-3-030-81508-0\_16}
  {\path{doi:10.1007/978-3-030-81508-0\_16}}.

\bibitem{HopUll79}
John~E. Hopcroft and Jeff~D. Ullman.
\newblock {\em Introduction to Automata Theory, Languages, and Computation}.
\newblock Addison-Wesley Publishing Company, 1979.

\bibitem{JKari01b}
Jarkko Kari.
\newblock Synchronizing finite automata on eulerian digraphs.
\newblock In {\em Math. Foundations Comput.Sci.; 26th Internat. Symp.,
  Marianske Lazne}, number 2136 in LNCS, pages 432--438, Berlin, Heidelberg,
  2001. Springer.

\bibitem{Kfoury70}
Denis~J. Kfoury.
\newblock Synchronizing sequences for probabilistic automata.
\newblock {\em Stud. Appl. Math.}, 49:101--103, 1970.

\bibitem{Kohavi70}
Zvi Kohavi.
\newblock {\em Switching and finite automata theory}.
\newblock McGraw Hill, New-York, 1970.

\bibitem{Mar2009b}
Pavel Martyugin.
\newblock Complexity of problems concerning reset words for some partial cases
  of automata.
\newblock {\em Acta Cybernetica}, 19(2):517--536, 2009.

\bibitem{DBLP:conf/focs/Natarajan86}
Balas~K. Natarajan.
\newblock An algorithmic approach to the automated design of parts orienters.
\newblock In {\em 27th Annual Symposium on Foundations of Computer Science,
  Toronto, Canada, 27-29 October 1986}, pages 132--142. {IEEE} Computer
  Society, 1986.

\bibitem{DBLP:journals/ijrr/Natarajan89}
Balas~K. Natarajan.
\newblock Some paradigms for the automated design of parts feeders.
\newblock {\em Int. J. Robotics Res.}, 8(6):98--109, 1989.

\bibitem{Pin77a}
Jean{-}\'{E}ric Pin.
\newblock Sur la longueur des mots de rang donn\'e d'un automate fini.
\newblock {\em C. R. Acad. Sci. Paris S\'er. A-B}, 284:1233--1235, 1977.

\bibitem{Pin78a}
Jean{-}\'{E}ric Pin.
\newblock Sur un cas particulier de la conjecture de {\v c}ern{\'y}.
\newblock In {\em 5th ICALP}, number~62 in LNCS, pages 345--352, Berlin, 1978.
  Springer.

\bibitem{RamadgeWonham87}
Peter~J. Ramadge and Walter~Murray Wonham.
\newblock Supervisory control of a class of discrete event processes.
\newblock {\em SIAM Journal on Control and Optimization}, 25:206--230, 1987.

\bibitem{Rystsov1995b}
Igor~K. Rystsov.
\newblock Almost optimal bound of recurrent word length for regular automata.
\newblock {\em Cybernetics and Systems Analysis volume}, 31:669--674, 1995.

\bibitem{Rystsov1995a}
Igor~K. Rystsov.
\newblock Quasioptimal bound for the length of reset words for regular
  automata.
\newblock {\em Acta Cybernetic}, 12(2):145--152, 1995.

\bibitem{Rys96}
Igor~K. Rystsov.
\newblock Exact linear bound for the length of reset words in commutative
  automata.
\newblock {\em Publicationes Mathematicae, Debrecen}, 48(3-4):405--409, 1996.

\bibitem{Rystsov97}
Igor~K. Rystsov.
\newblock Reset words for commutative and solvable automata.
\newblock {\em Theoret. Comput. Sci.}, 172:273--279, 1997.

\bibitem{Rystsov2000}
Igor~K. Rystsov.
\newblock Estimation of the length of reset words for automata with simple
  idempotents.
\newblock {\em Cybernetics and Systems Analysis}, 36(3):339--344, May 2000.

\bibitem{DBLP:journals/tcs/Ryzhikov19a}
Andrew Ryzhikov.
\newblock Synchronization problems in automata without non-trivial cycles.
\newblock {\em Theor. Comput. Sci.}, 787:77--88, 2019.
\newblock \href {https://doi.org/10.1016/j.tcs.2018.12.026}
  {\path{doi:10.1016/j.tcs.2018.12.026}}.

\bibitem{San2005}
Sven Sandberg.
\newblock Homing and synchronizing sequences.
\newblock In M.~Broy, B.~Jonsson, J.-P. Katoen, M.~Leucker, and A.~Pretschner,
  editors, {\em Model-Based Testing of Reactive Systems}, volume 3472 of {\em
  Lect. Notes Comp. Sci.}, pages 5--33. Springer, 2005.

\bibitem{shitov2019}
Yaroslav Shitov.
\newblock An improvement to a recent upper bound for synchronizing words of
  finite automata.
\newblock {\em Journal of Automata, Languages and Combinatorics},
  24(2--4):367--373, 2019.

\bibitem{Starke66}
Peter~H. Starke.
\newblock Eine {B}emerkung \"uber homogene {E}xperimente.
\newblock {\em Elektronische Informationverarbeitung und Kybernetik (later
  Journal of Information Processing and Cybernetics)}, 2:61--82, 1966.

\bibitem{Steinberg11a}
Benjamin Steinberg.
\newblock The{ \v{C}}ern{\'{y}} conjecture for one-cluster automata with prime
  length cycle.
\newblock {\em Theor. Comput. Sci.}, 412(39):5487--5491, 2011.

\bibitem{szykula2018}
Marek Szykula.
\newblock {Improving the Upper Bound on the Length of the Shortest Reset Word}.
\newblock In Rolf Niedermeier and Brigitte Vall{\'e}e, editors, {\em 35th
  Symposium on Theoretical Aspects of Computer Science (STACS 2018)}, volume~96
  of {\em Leibniz International Proceedings in Informatics (LIPIcs)}, pages
  56:1--56:13, Dagstuhl, Germany, 2018. Schloss Dagstuhl--Leibniz-Zentrum fuer
  Informatik.

\bibitem{Trahtman07}
Avraham~N. Trahtman.
\newblock The{ \v{C}ern\'y} conjecture for aperiodic automata.
\newblock {\em Discrete Mathematics \& Theoretical Computer Science},
  9(2):3--10, 2007.

\bibitem{Vol2008}
Mikhail~V. Volkov.
\newblock Synchronizing automata and the {{\v C}}ern\'y conjecture.
\newblock In Carlos Mart\'{\i}n-Vide, Friedrich Otto, and Henning Fernau,
  editors, {\em Language and Automata Theory and Applications, 2nd Int.
  Conference, LATA}, volume 5196 of {\em Lect. Notes Comp. Sci.}, pages 11--27.
  Springer, 2008.

\bibitem{Volkov09}
Mikhail~V. Volkov.
\newblock Synchronizing automata preserving a chain of partial orders.
\newblock {\em Theor. Comput. Sci.}, 410(37):3513--3519, 2009.

\bibitem{wonham2019}
Walter~Murray Wonham and Kai Cai.
\newblock {\em Supervisory Control of Discrete-Event Systems}.
\newblock Springer, 2019.

\end{thebibliography}
\end{document}